\documentclass[letterpaper, 11pt]{article}
\usepackage{amsmath}
\usepackage{amssymb}
\usepackage{amsthm}
\usepackage[linesnumbered,ruled]{algorithm2e}
\usepackage{todonotes}
\usepackage{enumitem}
\usepackage{times,amsxtra,latexsym,graphics,verbatim,epsfig,setspace,subfigure,epic}
\usepackage{xspace,prettyref}
\usepackage{bbm}
\usepackage{authblk}

\usepackage{tikz}
\usetikzlibrary{shapes,snakes}

\DeclareMathOperator*{\argmin}{argmin}
\DeclareMathOperator*{\argmax}{argmax}

\newtheorem{definition}{Definition}[section]
\newtheorem{theorem}{Theorem}[section]

\newtheorem{proposition}[theorem]{Proposition}
\newtheorem{corollary}[theorem]{Corollary}
\newtheorem{claim}[theorem]{Claim}
\newtheorem{lemma}[theorem]{Lemma}

\newcommand{\naturals}{{\mathbb{N}}}

\newrefformat{eq}{(\ref{#1})}
\newrefformat{chap}{Chapter~\ref{#1}}
\newrefformat{sec}{Section~\ref{#1}}
\newrefformat{algo}{Algorithm~\ref{#1}}
\newrefformat{fig}{Fig.~\ref{#1}}
\newrefformat{tab}{Table~\ref{#1}}
\newrefformat{rmk}{Remark~\ref{#1}}
\newrefformat{clm}{Claim~\ref{#1}}
\newrefformat{def}{Definition~\ref{#1}}
\newrefformat{cor}{Corollary~\ref{#1}}
\newrefformat{lmm}{Lemma~\ref{#1}}
\newrefformat{prop}{Proposition~\ref{#1}}
\newrefformat{app}{Appendix~\ref{#1}}
\newrefformat{hyp}{Hypothesis~\ref{#1}}
\newrefformat{thm}{Theorem~\ref{#1}}

\newcommand{\pth}[1]{\left( #1 \right)}

\newcommand{\calE}{{\mathcal{E}}}
\newcommand{\calF}{{\mathcal{F}}}

\newcommand{\calL}{{\mathcal{L}}}
\newcommand{\calM}{{\mathcal{M}}}

\newcommand{\calP}{{\mathcal{P}}}

\newcommand{\calS}{{\mathcal{S}}}
\newcommand{\calT}{{\mathcal{T}}}

\newcommand{\calV}{{\mathcal{V}}}

\title{Reaching Approximate Byzantine Consensus with Multi-hop Communication\thanks{This research is supported in part by National Science Foundation awards NSF 1329681.
Any opinions, findings, and conclusions or recommendations expressed here are those of the authors
and do not necessarily reflect the views of the funding agencies or the U.S. government.}
}

\author{Lili Su\thanks{Address: 459 Coordinated Science Lab MC 228, 1308 W. Main St.
Urbana Illinois 61801, USA. \\
Email: lilisu3@illinois.edu. \\
Tel.: +12178985405
}}
\author{Nitin Vaidya}
\affil{ Department of Electrical and Computer Engineering\\ University of Illinois at Urbana-Champaign}

\begin{document}
\date{~}

\maketitle


~
\begin{abstract}
We address the problem of reaching consensus in the presence of Byzantine faults. In particular, we are interested in investigating the impact of messages relay on the network connectivity for a correct iterative approximate Byzantine consensus algorithm to exist.
The network is modeled by a simple directed graph. We assume a node can send messages to another node that is up to $l$ hops away via forwarding by the intermediate nodes on the routes,
where $l\in \naturals$ is a natural number. 
%
We characterize the necessary and sufficient topological conditions on the network structure. The tight conditions we found are consistent with the tight conditions identified in \cite{Vaidya2012IABC} for $l=1$, where only local communication is allowed, and are strictly weaker
for $l>1$. Let $l^*$ denote the length of a longest path in the given network.  For $l\ge l^*$ and undirected graphs, our conditions hold if and only if  $n\ge 3f+1$ and the node-connectivity of the given graph is at least $2f+1$ , where $n$ is the total number of nodes and $f$ is the maximal number of Byzantine nodes; and for $l\ge l^*$ and directed graphs, our conditions is equivalent to the tight condition found in \cite{Tseng2014}, wherein exact Byzantine consensus is considered.

Our sufficiency is shown by constructing a correct algorithm, wherein the trim function is constructed based on investigating a newly introduced minimal messages cover property.
The trim function proposed also works over multi-graphs.

\end{abstract}



\newpage
\section{Introduction}
\label{sec:intro}



 Reaching consensus resiliently in the presence of Byzantine faults has been studied extensively in distributed computing \cite{Lamport1982, Rabin1983, Ben-Or1983, AA_Fekete_aoptimal,Ben-Or2010}. 
Messages relay is the relaying of a message from its source toward its ultimate destination through intermediate nodes.
We say a messages relay is bounded if each source node can only reliably/noislessly send messages to a destination node that is up to $l$ hops away, where $l\in \naturals$, termed as relay depth. Our focus is on investigating the tradeoff between the relay depth $l$ and the network connectivity for a correct iterative approximate Byzantine consensus algorithm to exist. Let $l^*$ be the length of a longest path in the network. The two special cases, $l\ge l^*$ and $l=1$, respectively, have already been well studied.

%
Under the full forwarding model, i.e., $l\ge l^*$, a node is able to reliably send messages to another node via every possible route in the network. Let $n$ be the total number of nodes in the network, it has been shown that given $f$ Byzantine nodes, if the network node-connectivity is at least $2f+1$ and $n\ge 3f+1$, there exist algorithmic solutions for the fault-free nodes to reach consensus over all possible inputs. Conversely, if the network node-connectivity is strictly less than $2f+1$ or $n<3f+1$, then reaching consensus is not guaranteed \cite{impossible_proof_lynch}. Thus $2f+1$ node-connectivity and $n\ge 3f+1$ are both necessary and sufficient.
However,  as a result of this communication assumption, the proposed algorithms require fault-free nodes to keep track of the \emph{entire} network topology, leading to huge consumption of both memory resource and computation power.
In contrast, iterative algorithms are typically characterized by local communication (among neighbors, or near-neighbors), simple computations performed repeatedly, and a small amount of state per node.
The purely local communication model (i.e., $l=1$), where a node can only send messages to its neighbors and no message forwarding is allowed, has also attracted extensive attention among researchers \cite{Benezit, leblanc_HiCoNs,jadbabaie_concensus,Vaidya2012IABC,Vaidyamatrix, vaidyaII}.
It has been shown that a correct iterative approximate Byzantine algorithm exists if and only if for any node partition $L, C, R, F$ of a graph such that $L\not=\O$, $R\not=\O$ and $|F|\le f$, either there exists a node $i\in L$ such that $|N_i^{-}\cap (R\cup C)|\ge f+1$ or there exists a node $i\in R$ such that $|N_i^{-}\cap (L\cup C)|\ge f+1$, where $N_i^{-}$ is the collection of incoming neighbors of node $i$.

 Our main contribution is to provide a family of tight sufficient and necessary conditions on the network topology for a correct iterative consensus algorithm to exist. Our sufficiency is proved by constructing a new simple iterative algorithm, whose trim function is based on investigating a newly introduced minimal messages cover property.
 Our results bridge the existing aforementioned two streams of work, i.e., when $l\ge l^*$ and $l=1$, respectively, and fill the gap between these two models.

The rest of the paper is  organized as follows. Section \ref{sec: model} presents our models and the structure of iterative algorithms of interest. Our necessary condition is demonstrated in Section \ref{sec:necessary}, whose sufficiency is proved constructively in Section \ref{sec:sufficiency}. We shown in Section \ref{sec: extension} that our results are equivalent to the $2f+1$ node-connectivity and $n\ge 3f+1$ conditions for undirected graph when $l\ge l^*$. Section \ref{sec:conclusion} discusses possible relaxations of our fault model and concludes the paper.

\section{Problem setup and structure of iterative algorithms}
\label{sec: model}
%
%
\paragraph{Communication model}
The system is assumed to be {\em synchronous}.
The communication network is modeled as a simple {\em directed} graph $G$. Define two functions $\calV(\cdot)$ and $\calE(\cdot)$ over a graph $G$ as follows:  $\calV(G)=\{1,\dots,n\}$ returns the set of $n$ nodes, where $n\ge 2$, and $\calE(G)$ returns the set of directed edges between nodes in $\calV(G)$. 
Node $i$ can send messages to node $j$ if and only if there exists an $i,j$--path of length at most $l$ in $G$, where $l\in \naturals$ is a natural number.
In addition, we assume each node can send messages to itself as well, i.e., $(i,i)\in\calE(G)$ for all $i\in\calV(G)$. 
For each node $i$, let $N_i^{l-}$ be the set of nodes that can reach node $i$ via at most $l$ hops.
Similarly, denote the set of nodes that are reachable from node $i$ via at most $l$ hops by $N_i^{l+}$. Due to the existence of self-loops, $i\in N_i^{l-}$ and $i\in N_i^{l+}$. When $l=1$, we write $N_i^{1-}$ and $N_i^{1+}$ as $N_i^{-}$ and $N_i^{+}$, respectively, for simplicity.
Note that node $i$ may send a message to node $j$ via different $i, j$--paths. To capture this distinction in transmission routes, we represent a message as a tuple $m=(w, P)$, where $w\in \mathbb{R}$ and $P$ indicates the path via which message $m$ should be transmitted. Four functions are defined over $m$. Let function $\mathsf{value}$ be $\mathsf{value}(m)=w$ and let  $\mathsf{path}$ be $\mathsf{path}(m)=P$, whose images are the first entry and the second entry, respectively, of message $m$. 
In addition, functions $\mathsf{source}$ and $\mathsf{destination}$ are defined by  $\mathsf{source}(m)=i$ and $\mathsf{destination}(m)=j$ if  $P$ is an $i, j$--path, i.e., message $m$ is sent from node $i$ to node $j$.




\paragraph{Fault model}
Let $\calF\subseteq \calV(G)$ be the collection of faulty nodes in the system. We consider the Byzantine fault model with up to $f$ nodes becoming faulty, i.e., $|\calF|\le f$.  A faulty node may {\em misbehave} arbitrarily. Possible misbehavior includes sending incorrect and mismatching (or inconsistent) messages to different neighbors. In addition, a faulty node $k\in \calF$ may tamper message $m$ if it is in the transmission path, i.e., $k\in \calV(\mathsf{path}(m))$\footnote{Recall that $\calV(\cdot)$ is the vertex set of a given graph and $\calV(\mathsf{path}(m))$ denotes the collection of vertices along the route of message $m$, including the source and the

  destination.}.
However, faulty nodes are only able to tamper $\mathsf{value}(m)$, leaving $\mathsf{path}(m)$ unchanged. This assumption is placed for ease of exposition, later in Section \ref{sec:conclusion} we relax this assumption by considering the possibilities that faulty nodes may also tamper messages paths or even fake and transmit non-existing messages.
%
 Faulty nodes are also assumed to have complete knowledge of the execution of the algorithm, including the states of all nodes,
contents of messages the other nodes send to each other, and
the algorithm specification, so that they may potentially collaborate with each other adaptively.


\paragraph{Iterative approximate Byzantine consensus (IABC) algorithms}
\label{sec:iabc}

The iterative algorithms considered in this paper should have the following structure:
Each node $i$ maintains state $v_i$, with $v_i[t]$ denoting the state
of node $i$ at the {\em end}\, of the $t$-th iteration of the algorithm.
Initial state of node $i$,
$v_i[0]$, is equal to the initial {\em input}\, provided to node $i$.
At the {\em start} of the $t$-th iteration ($t>0$), the state of
node $i$ is $v_i[t-1]$.
The IABC algorithms of interest will require each node $i$
to perform the following three steps in iteration $t$, where $t>0$.
Note that the faulty nodes may deviate from this specification.

\begin{enumerate}
\item {\em Transmit step:} Transmit messages of the form $(v_i[t-1], \cdot)$ to nodes in $N_i^{l+}$, i.e., the nodes that are reachable from node $i$ via at most $l$ hops. If node $i$ is an intermediate node on the route of some message, then node $i$ forwards that message as instructed by the message path.
\item {\em Receive step:} Receive messages from $N_i^{l-}$, i.e., the nodes that can reach node $i$ via at most $l$ hops.
Denote by $\calM_i[t]$ the set of messages that node $i$ received
at iteration $t$. 
\item \textit{Update step:} Node $i$ updates its state using a transition function $Z_i$, where $Z_i$ is a part of the specification of the algorithm, and takes as input the set $\calM_i[t]$.
\begin{eqnarray}
v_i[t]=Z_i(\calM_i[t]).
\label{eq:Z_i}
\end{eqnarray}
\end{enumerate}
Note that at the $t$--th iteration, between step two and step three, by sending message to itself node $i$ is able to memorize its state in the immediate preceding iteration, i.e,. $v_i[t-1]$. However, at the end of update step, except for its updated state $v_i[t]$, no other information collected in current iteration or any of the previous iteration will be kept by node $i$. In step three, in general, $Z_i$ is some trim function over the received messages collection $\calM_i[t]$. The trimming strategy may depends on message values, message routes, or both. In addition, different nodes are allowed to have different trimming strategies.

Let $U[t]$ be the largest state among the fault-free nodes at the end of the $t$-th iteration, i.e., $U[t] = \max_{i\in\calV-\calF}\,v_i[t]$.
Since the initial state of each node is equal to its input,
$U[0]$ is equal to the maximum value of the initial input at the fault-free nodes. Similarly, we define $\mu[t]$ to be the smallest state at the $t$--th iteration and $\mu[0]$ to be the smallest initial input.
%
%
%
For an IABC algorithm to be correct, the following two conditions must be satisfied:
\begin{itemize}
\item {\em Validity:} $\forall~~ t>0,
~~\mu[t]\ge \mu[0]
~\mbox{~~and~~}~
~U[t]\le U[0]$

\item {\em Convergence:} $\lim_{\,t\rightarrow\infty} ~ U[t]-\mu[t] = 0$
\end{itemize}

Our focus is to identify the necessary and sufficient
conditions for the existence of a {\em correct} IABC algorithm (i.e.,
an algorithm satisfying the above validity and convergence conditions)
for a given $G$ and a given $l$.
%


\section{Necessary Conditions}
\label{sec:necessary}

For a correct IABC algorithm to exist, the underlying network $G$ must satisfy the conditions presented in this section. A couple of definitions are needed before we are able to formally state our necessary conditions.

\begin{definition}
Let $W$ be a set of vertices in $G$ and $x$ be a vertex in $G$ such that $x\notin W$. A $W, x$--path is a path from some vertex $w\in W$ to vertex $x$. A set $S$ of vertices such that $x\notin S$ is a $W,x$--\emph{vertex cut} if every $W,x$--path contains a vertex in $S$. The minimum size of a $W,x$--vertex cut is called the $W,x$--connectivity and is denoted by $\kappa(W,x)$. Similarly, a set $S_l$ of vertices is an $l$--\emph{restricted vertex cut} if the deletion of $S_l$ destroys all $W,x$--paths of length at most $l$. Let $\kappa_l(W,x)$ be the minimum size of such restricted vertex cut in $G$.
\end{definition}
The first part of the above definition is the classical definition of node connectivity in graph theory. However, this definition is a global notion. In our communication model, we implicitly assume that each fault-free node only knows the local network topology up to its $l$--th neighborhood. We adapt node connectivity to our model by restricting the length of the paths under consideration. Note that $\kappa_l(W,x)=\kappa(W,x)$ for all $l\ge l^*$, and that a $1$--restricted vertex cut of $(W, x)$ is the number of node $x$'s incoming neighbors in $W$, i.e., $\kappa_1(W,x)=\left|W\cap N_{x}^{-}\right|$.

\begin{definition}
For non-empty disjoint sets of nodes $A$ and $B$ in $G$, we say
 $A\Rightarrow_l B$ if and only if
 there exists a node $i\in B$ such that $\kappa_{l}(A, i)\ge f+1$;
$A\nRightarrow_l B$ otherwise.
\end{definition}

Let $F\subseteq \calV(G)$ be a set of vertices in $G$, denote the induced subgraph\footnote{An induced subgraph of $G$, induced by vertex set $S\subseteq \calV(G)$, is the subgraph $H$ with vertex set $S$ such that $\calE(H)=\{(u,v)\in \calE(G): ~ u, v\in S\}$. Recall that $\calV(\cdot)$ and $\calE(\cdot)$ are the vertex set and edge set, respectively, of a given graph.} of $G$ induced by vertex set $\calV(G)-F$ by $G_F$. We describe the necessary and sufficient condition below, whose necessity is proved in Theorem \ref{thm:nc} and sufficiency is shown constructively in Section \ref{sec:sufficiency}. For ease of future reference, we termed the condition as \emph{Condition NC}.

\noindent\textbf{Condition NC}:
For any node partition $L, C, R, F$ of $G$ such that $L\not=\O, R\not=\O$ and $|F|\le f$, in the induced subgraph $G_F$,
at least one of the two conditions below must be true: (i) $R\cup C\Rightarrow_l L$; (ii) $L\cup C\Rightarrow_l R$. 

Intuitively, Condition NC requires that either the set of nodes in $R\cup C$ are able to collectively influence the state of a node in $L$ or vice versa. Note that when $l=1$, Condition NC becomes\\
\emph{``
For any node partition $L, C, R, F$ of $G$ such that $L\not=\O, R\not=\O$ and $|F|\le f$, in the induced subgraph $G_F$,
at least one of the two conditions below must be true: (i) there exists a node $i\in L$ such that $\left |\pth{R\cup C}\cap N_i^{-}\right |\ge f+1$; (ii) there exists a node $j\in R$ such that $\left |\pth{L\cup C}\cap N_j^{-}\right |\ge f+1$."}, which is shown to be both necessary and sufficient without message relay in \cite{Vaidya2012IABC}.

\begin{theorem}
\label{thm:nc}
Suppose that a correct IABC algorithm exists for $G$. Then $G$ satisfies Condition NC.
\end{theorem}
We prove this theorem in Appendix \ref{app:necessary}. Our proof shares the same proof structure of Theorem 1 in \cite{Vaidya2012IABC}.
The basic idea is as follows: Suppose there exists a correct IABC algorithm, then we are able to find a node partition satisfying the conditions as listed in Condition NC,
such that under some Byzantine layout, and for some specific initial inputs, convergence condition will be violated.

The above necessary condition is in general stronger than the necessary condition derived under single-hop message transmission model (i.e., $l=1$) \cite{Vaidya2012IABC}.
Consider the system depicted in Fig. \ref{example}.  The topology of this system does not satisfy the necessary condition derived in \cite{Vaidya2012IABC}. Since in the node partition $L=\{p_1, p_4\}, R=\{p_2, p_3\}, C=\O$ and $F=\{p_5\}$, neither $L\cup C\Rightarrow_l R$ nor $R\cup C\Rightarrow_l L$ holds for $l=1$. However, via enumeration it can be seen that the above graph (depicted in Fig. \ref{example}) satisfies Condition NC when $l\ge 2$. Nevertheless, increasing relay depth does not always admit more graph structures. For instance, for $n=4$, $f=1$ and any $l$, the only graph that satisfy Condition NC is the complete graph.

It follows from the definition of Condition NC that if a graph $G$ satisfies Condition NC for $l\in \{1,\ldots, n-1\}$, then $G$ also satisfies Condition NC for all $l^{\prime}\ge l$. Let $l_0$ be the smallest integer for which $G$ satisfies Condition NC, where $l_0=n$ by convention if $G$ does not satisfy Condition NC for any $l\in \{1,\ldots, n-1\}$. We observe that in general given a graph $G$, the diameter of $G$ can be arbitrarily smaller than $l_0$. For instance, the diameter of the graph depicted in Fig. \ref{example2} is two. However, for the depicted graph, $l_0=\frac{n+1}{4}$ when $\frac{n-1}{2}$ is odd. So $l_0$ is much larger than two for large $n$.

 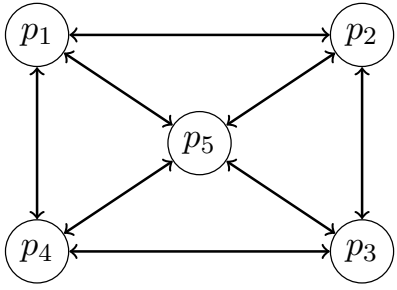
\begin{figure}

 \centering
    \scalebox{1.2}{

    \begin{tikzpicture}[auto, scale=1.2]

     \node[draw,circle,minimum size=0.7cm,inner sep=0pt] (A) at (0,0)  {$p_1$ };
     \node [draw,circle,minimum size=0.7cm,inner sep=0pt] (D) at (0,-2) {$p_4$};
      \node [draw,circle,minimum size=0.7cm,inner sep=0pt] (B) at (3,0) {$p_2$};
       \node [draw,circle,minimum size=0.7cm,inner sep=0pt] (C) at (3,-2) {$p_3$};
       \node [draw,circle,minimum size=0.7cm,inner sep=0pt] (E) at (1.5,-1) {$p_5$};

      \draw[<->, thick] (A)--(B);
      \draw[<->, thick] (B) -- (C);
      \draw[<->, thick] (D) -- (C);
      \draw[<->, thick] (A) -- (D);

      \draw[<->, thick] (A) -- (E);
      \draw[<->, thick] (B) -- (E);
      \draw[<->, thick] (C) -- (E);
      \draw[<->, thick] (D) -- (E);
    \end{tikzpicture}
    }
    \caption{In this system, there are five processors $p_1, p_2, p_3, p_4$ and $p_5$; all communication links are bi-directional; and at most one processor can be adversarial, i.e., $f=1$. }
    \label{example}
    \end{figure}

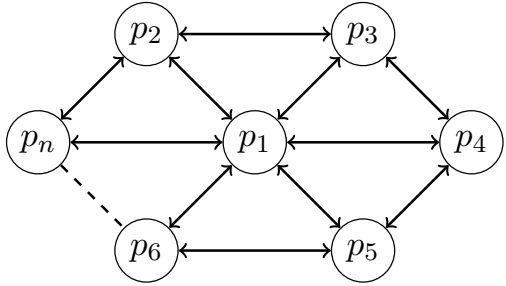
\begin{figure}

 \centering
    \scalebox{1.2}{

    \begin{tikzpicture}[auto, scale=1.2]

     \node[draw,circle,minimum size=0.7cm,inner sep=0pt] (A) at (0,0)  {$p_2$ };
     \node [draw,circle,minimum size=0.7cm,inner sep=0pt] (B) at (-1,-1) {$p_n$};
     \node [draw,circle,minimum size=0.7cm,inner sep=0pt] (C) at (0,-2) {$p_6$};

      \node [draw,circle,minimum size=0.7cm,inner sep=0pt] (D) at (2,-2) {$p_5$};
       \node [draw,circle,minimum size=0.7cm,inner sep=0pt] (E) at (3,-1) {$p_4$};
        \node[draw,circle,minimum size=0.7cm,inner sep=0pt] (F) at (2,0)  {$p_3$ };
\node [draw,circle,minimum size=0.7cm,inner sep=0pt] (G) at (1,-1) {$p_1$};
      \draw[<->, thick] (A)--(B);
      \draw[-, dashed, thick] (B) -- (C);
      \draw[<->, thick] (D) -- (C);
      \draw[<->, thick] (D) -- (E);
      \draw[<->, thick] (E) -- (F);
      \draw[<->, thick] (F) -- (A);
      \draw[<->, thick] (A) -- (G);
      \draw[<->, thick] (B) -- (G);
      \draw[<->, thick] (C) -- (G);
      \draw[<->, thick] (D) -- (G);
      \draw[<->, thick] (E) -- (G);
      \draw[<->, thick] (F) -- (G);
    \end{tikzpicture}
    }
    \caption{In this system, there are $n$ processors $p_1, \ldots, p_n$; all communication links are bi-directional; and at most one processor can be adversarial, i.e., $f=1$. Nodes $p_2, \ldots, p_n$ form a cycle of length $n-1$ and these nodes are all connected to node $p_1$. }
    \label{example2}
    \end{figure}

Similar to \cite{Vaidya2012IABC}, as stated in our next corollary, our Condition NC for general $l$ also implies a lower bound on both the graph size $n$ and the incoming degree of each node. Moreover, this lower bound is independent of $l$.
\begin{corollary}
\label{cdegree}
If $G$ satisfies Condition NC, then
 $n$ must be at least $3f+1$, and each node must have at least $2f+1$ incoming neighbors other than itself, i.e., $|N_i^{-}-\{i\}|\ge 2f+1$.
\end{corollary}
The proof of Corollary \ref{cdegree} can be found in Appendix \ref{app:cdegree}. Note that Corollary \ref{cdegree} also characterizes a lower bound on the density of $G$, that is $\left |\calE(G)\right |\ge n(2f+2)$, including self-loops, which is independent of the relay depth $l$ as well. Proposition \ref{density} says that for $f=1$, communication over multi-hop does not imply the existence of a sparser graph for which Condition NC holds than that with communication over single-hop. For $f>1$ whether there exists a graph satisfying Condition NC with $N_i^{-}=2f+2$ and $(2f+2)n$ edges or not for any $l$ is still open.

\begin{proposition}
\label{density}
For $f=1$ and $l_0=1$, there exists a graph $G$ for any $n\ge 3f+1=4$ such that
(i) $|N_i^{-}|=4$ for all $i\in \calV(G)$; and
(ii) $|\calE(G)|=4n$.
\end{proposition}

\subsection{Equivalent Characterization of Condition NC}
Informally speaking, Condition NC describes the information propagation property in terms of four sets partitions. In this subsection, an equivalent condition of Condition NC is proposed, which is based on characterizing the structure of the special subgraphs, termed as reduced graph, of the power graph $G^l$. The new condition suggests that all fault-free nodes will be influence by a collection of common fault-free nodes.

\begin{definition}
\label{def:decompose}
{\bf Meta-graph of SCCs:}
Let $K_1,K_2,\ldots,K_k$ be the strongly connected components (i.e., SCCs) of $G$. The graph of SCCs, denoted by $G^{SCC}$, is defined by\\
(i) Nodes are $K_1,K_2,\ldots,K_k$; and\\
(ii) there is an edge $(K_i, K_j)$ if there is some $u\in K_i$ and $v\in K_j$ such that $(u, v)$ is an edge in $G$.\\
Strongly connected component $K_h$ is said to be a {\em source component}
if the corresponding node in $G^{SCC}$ is \underline{not} reachable from any
other node in $G^{SCC}$.
\end{definition}
It is known that the $G^{SCC}$ is a directed acyclic graph ( i.e., DAG ) \cite{dag_decomposition}, which contains no directed cycles. It can be easily checked that due to the absence of directed cycles and finiteness, there exists one node in $G^{SCC}$ that is not reachable from any other node. That is, a graph $G$ has at least one source component.
\begin{definition}
The $l$--th power of a graph $G$, denoted by $G^l$,  is a graph with the same set of vertices as $G$ and a directed edge between two vertices $u, v$ if and only if there is a path of length $l$ from $u$ to $v$ in $G$.
\end{definition}

A path of length one between vertices $u$ and $v$ in $G$ exists if $(u, v)$ is an edge in $G$. And a path of length two between vertices $u$ and $v$ in $G$ exists for every vertex $w$ such that $(u,w)$ and $(w,v)$ are edges in $G$.
Then for a given graph $G$ with self-loop at each node,
the $(u,v)^{th}$ element in the square of the adjacency matrix of $G$ counts the number of paths of length at most two in $G$. Similarly, the $(u,v)^{th}$ element in the $l$--th power of the adjacency matrix of $G$ gives the number of paths of length at most $l$ between vertices $u$ and $v$ in $G$.
The power graph $G^{l}$ is a multigraph\footnote{A multigraph (or pseudograph) is a graph which is permitted to have multiple edges between each vertex pair, that is, edges that have the same end nodes. Thus two vertices may be connected by more than one edge.} and there is a one-to-one correspondence between an edge $e$ in $G^{l}$ and a path of length at most $l$ in $G$. Let $e$ be an edge in $G^l$, and let $P(e)$ be the corresponding path in $G$, we say an edge $e$ in $G^{l}$ is covered by node set $S$, if $\calV(P(e))\cap S\not=\O$, i.e., path $P(e)$ passes through a node in $S$.


\begin{definition}[Reduced Graph]
\label{def:reduced} 

For a given graph $G$ and $F\subseteq\calV(G)$, let $E=\{e\in \calE(G^{l}):~\calV(P(e))\cap F\not=\O\}$ be the set of edges in $G^{l}$ that are covered by node set $F$. For each node $i\in \calV(G)-F$, choose $C_i\subseteq N_i^{l-}-\{i\}$ such that $|C_i|\le f$. Let $$E_i=\{e\in \calE(G^{l}): e~\text{is an incoming edge of node $i$ in}~G^{l}~~ \text{and}~ \calV(P(e))\cap C_i\not=\O\}$$ be the set of incoming edges of node $i$ in $G^{l}$ that are covered by node set $C_i$.
  A reduced graph of $G^l$, denoted by $\widetilde{G^l}_F$, is a subgraph of $G^l$ whose node set and edge set are defined by (i)~ $\calV(\widetilde{G^l_F})=\calV(G)-F$; and
(ii) $\calE(\widetilde{G^l_F})=\calE(G^{l})-E-\cup_{i\in \calV(G)-F}E_i$, respectively.
\end{definition}

Note that for a given $G$ and a given $F$,
multiple reduced graphs may exist.
Let us define set $R_F$ to be the collection of all reduced graph of $G^l$ for a given $F$, i.e.,
\begin{eqnarray}
R_F & = & \{ \widetilde{G^l}_F: ~ \widetilde{G^l}_F ~~\text{is a reduced graph of}~~ G^l \}.
\end{eqnarray}

Since $G_F^l$, the $l$--th power of the induced subgraph $G_F$, itself is a reduced graph of $G^l$, where we choose $C_i=\O$ for each $i\in \calV(G)-F$, thus $R_F$ is nonempty.
 In addition, $|R_F|$ is finite since the graph $G$ is finite,

\begin{theorem}
\label{thm:nc2}
Graph $G$ satisfies Condition NC if and only if 
every reduced graph $\widetilde{G^l}_F$ obtained as per Definition \ref{def:reduced}
must contain exactly one {\em source component}.
\end{theorem}

\section{Sufficiency: Algorithm 1}
\label{sec:sufficiency}

As aforementioned, for each node $i$, the collection of received messages $\calM_i[t]$ may contains bogus messages and/or tampered messages due to the existence of Byzantine nodes, thus $Z_i(\cdot)$ is in general a trimming function. In this section we propose an algorithm, termed Algorithm 1, using a novel update/trimming strategy and show its correctness. First we introduce the definition of message cover that will be used frequently in this section. 

\begin{definition}

  For a communication graph $G$, let $\calM$ be a set of messages, and let $\calP(\calM)$ be the set of paths corresponding to all the messages in $\calM$, i.e., $\calP(\calM)= \{\mathsf{path}(m)|m\in \calM\}$. A message cover of $\calM$ is a set of nodes $\calT(\calM)\subseteq \calV(G)$, such that for each path $P\in \calP$, we have $\calV(P)\cap \calT(\calM) \not=\O$. In particular, a minimum message cover is defined by
\begin{align*}
\calT^*(\calM)\in \argmin_{\calT(\calM)\subseteq\calV(G):~~ \calT(\calM)~ \text{is a cover of}~~ \calM} |\calT(\calM)|.
\end{align*}

 Conversely, given a set of messages $\calM_0$ and a set of nodes $\calT\subseteq \calV(G)$, a maximal set of messages $\calM\subseteq\calM_0$ that are covered by $\calT$ is defined by,
 \begin{align*}
 \calM^*\in \argmax_{\calM\subseteq\calM_0:~~\calT~\text{is a cover of}~\calM} |\calM|.
 \end{align*}
\end{definition}

We further need the following two definitions before we are able to proceed to the description of our algorithm. Recall that $\calM_i[t]$ is the collection of messages received by node $i$ at iteration $t$. Let $\calM_i^{\prime}[t]=\calM_i[t]-\{(v_i[t-1], (i,i))\}.$
Sort messages in $\calM_i^{\prime}[t]$ in an increasing order, according to their message values, i.e., $\mathsf{value}(m)$ for $m\in \calM_i^{\prime}[t]$.
Let $\calM_{is}[t]\subseteq \calM_i^{\prime}[t]$ such that
(i) for all $m\in \calM_i^{\prime}[t]-\calM_{is}[t]$ and $m^{\prime}\in \calM_{is}[t]$ we have $\mathsf{value}(m)\ge \mathsf{value}(m^{\prime})$;
and (ii) the cardinality of a minimum cover of $\calM_{is}[t]$ is exactly $f$, i.e., $|\calT^*(\calM_{is}[t])|=f$.
Similarly, we define $\calM_{il}[t]\subseteq \calM_i^{\prime}[t]$ as follows:
(i) for all $m\in \calM_i^{\prime}[t]-\calM_{il}[t]$ and $m^{\prime\prime}\in \calM_{il}[t]$ we have $\mathsf{value}(m)\le\mathsf{value}(m^{\prime\prime})$;
and (ii) the cardinality of a minimum cover of $\calM_{il}[t]$ is exactly $f$, i.e., $|\calT^*(\calM_{il}[t])|=f$.
In addition, define $\calM_{i}^{*}[t]=\calM_i^{\prime}[t]-\calM_{is}[t]-\calM_{il}[t]$.

\begin{theorem}
\label{trimming}
Suppose that graph $G$ satisfies Condition NC, then the sets of messages $\calM_{is}[t]$, $\calM_{il}[t]$ are well-defined and  $\calM_{i}^{*}[t]$ is nonempty.
\end{theorem}
This theorem is proved by construction, i.e., an algorithm is constructed to find the sets $\calM_{is}[t]$, $\calM_{il}[t]$ for a given $\calM_i^{\prime}$. Details of the algorithm and its correctness proof can be found in Appendix \ref{trimmingProof}. With this trimming strategy at hand, 
we will prove that there exists an IABC algorithm -- particularly
{\em Algorithm 1} below -- that satisfies
the {\em validity} and {\em convergence} conditions provided that the
graph $G$ satisfies Condition NC. This implies that Condition NC is also sufficient. {\em Algorithm 1} has the three-step structure described
in Section~\ref{sec:iabc}.

\vspace*{8pt}\hrule
{\,

\bf Algorithm 1}
\vspace*{4pt}\hrule

\begin{enumerate}
\label{algorithm}
\item {\em Transmit step:} Transmit messages of the form $(v_i[t-1], \cdot)$ to nodes in $N_i^{l+}$. If node $i$ is an intermediate node of some message, then node $i$ forwards that message as instructed by the message path.
    When node $i$ expects to receive a message from a path but does not receive the message, the message value is assumed to be equal to some default message.

\item {\em Receive step:} Receive messages from $N_i^{l-}$.

\item {\em Update step:}

Define
\begin{eqnarray}
v_i[t] ~ = ~ Z_i(\calM_{i}[t]) ~ = a_{i}v_i[t-1]+\sum_{m\in \calM_{i}^*[t]} a_{i} \, w_m.
\label{e_Z}
\end{eqnarray}
where $w_m=\mathsf{value}(m) ~~~\text{and}~~~  a_{i} = \frac{1}{\lvert \calM_{i}^*[t]\rvert+1}$.

\end{enumerate}
\hrule
\,

Recall $\calM_{i}^{*}[t]=\calM_i^{\prime}[t]-\calM_{is}[t]-\calM_{il}[t]$.
The ``weight'' of each term on the right-hand side of (\ref{e_Z}) is $a_i$, where $0<a_i\leq 1$, and these weights add to 1.
For future reference, let us define $\alpha$, which is used in Theorem \ref{claim_1}, as:
\begin{eqnarray}
\label{lowerbound}
\alpha=\min_{ i\in \calV-\calF} a_i.
\end{eqnarray}
In \emph {Algorithm 1}, each fault-free node $i$'s state, $v_i[t]$,  is updated as a convex combination of all the \emph{messages values} collected by node $i$ at round $t$. In particular, for each message $m\in \calM^{\prime}[t]$, its coefficient is $a_i$ if the message is in $\calM^*_i[t]$ or the message is sent via self-loop of node $i$; otherwise, the coefficient of $m$ is zero.
The update step in \emph {Algorithm 1} is a generalization of the update steps proposed in \cite{Vaidyamatrix,IBA_broadcast_Sundaram}, where the update summation is over all the incoming neighbors of node $i$ instead of over message routes. In \cite{Vaidyamatrix,IBA_broadcast_Sundaram}, only single-hop communication is allowed, i.e., $l=1$, and the fault-free node $i$ can receive only one message from its incoming neighbor. With multi-hop communication, fault-free node can possibly receive messages from a node via multiple routes. Our trim functions in \emph {Algorithm 1} take the possible multi-route messages into account. In fact, \emph {Algorithm 1} also works with multi-graphs.


\subsection{Matrix Representation of Algorithm 1}
\label{s_claim}

With our trimming function, the iterative update of the state
of a fault-free node $i$ admits a nice matrix representation of states evolution of fault-free nodes.  We use boldface upper case letters to denote matrices, rows of matrices, and their entries. For instance, $\bf{A}$ denotes a matrix, ${\bf A}_i$ denotes the $i$-th row of matrix $\bf{A}$, and ${\bf A}_{ij}$ denotes the element at the intersection of the $i$-th row and the $j$-th column of matrix $\bf{A}$. Some useful concepts and theorems are reviewed briefly in Appendix \ref{MatrixPreliminaries}.

\begin{definition}
\label{d_stochastic}
A vector is said to be {\em stochastic} if all the entries
of the vector are {\em non-negative}, and the entries add up to 1.
A matrix is said to be row stochastic if each row of the matrix is a
stochastic vector.
\end{definition}

Recall that $\calF$ is the set of faulty nodes and $|\calF|=\phi$.
Without loss of generality, suppose that nodes 1 through $(n-\phi)$ are
fault-free, and if $\phi>0$, nodes $(n-\phi+1)$ through $n$ are faulty.
Denote by ${\bf{v}}[0]\in \mathbb{R}^{n-\phi}$ the column vector consisting of the initial states of
all the {\em fault-free} nodes.
Denote by ${\bf{v}}[t]$, where $t\geq 1$, the column vector consisting of
the states of all the {\em fault-free} nodes
at the end of the $t$-th iteration, $t\geq 1$, where the $i$-th element
of vector ${\bf{v}}[t]$ is state $v_i[t]$. 
\begin{theorem}
\label{claim_1}
{
We can express the iterative update of the state
of a fault-free node $i$ $~~(1\leq i\leq n-\phi)$
performed in (\ref{e_Z}) using the matrix form in (\ref{e_matrix_i})
below,
where ${\bf{M}}_i[t]$ satisfies the four conditions listed below.
In addition to $t$, the row vector ${\bf{M}}_i[t]$
may depend on the state vector ${\bf{v}}[t-1]$ as well as the
behavior of the faulty
nodes in $\calF$. For simplicity, the notation ${\bf{M}}_i[t]$ does not
explicitly represent this dependence.
\begin{eqnarray}
v_i[t] & = & {\bf{M}}_i[t] ~ {{\bf{v}}}[t-1]
\label{e_matrix_i}
\end{eqnarray}
}
\begin{enumerate}
\item ${\bf{M}}_i[t]$ is a {\em stochastic} row vector of size $(n-\phi)$.
Thus,
${\bf{M}}_{ij}[t]\geq 0$, where $1\leq j\leq n-\phi$, and
\[
\sum_{1\leq j\leq n-\phi}~{\bf{M}}_{ij}[t] ~ = ~ 1
\]

\item ${\bf{M}}_{ii}[t]\ge a_i\ge \alpha$.

\item ${\bf{M}}_{ij}[t]$ is non-zero
only if~~there exists a message $m\in \calM_i[t]$ such that $\mathsf{source}(m)=j$ and $\mathsf{destination}(m)=i$.
\item For any $t\geq 1$, there exists a reduced graph $\widetilde{G^l}_{\calF}\in R_\calF$ with adjacent matrix ${\bf{H}}[t]$ such that
$\beta \, {\bf{H}}[t] ~ \leq ~  {{\bf{M}}[t]}$, where $
\beta$ is some constant $0<\beta\le 1$ to be specified in Claim \ref{claimw}.
%
%

\end{enumerate}
\end{theorem}
In Appendix \ref{app:claim_1}, we prove the correctness of Theorem \ref{claim_1} by constructing ${\bf{M}}_i[t]$
for $1\leq i\leq n-\phi$. Our proof follows the same line of analysis as in the proof of Claim 2 in \cite{Vaidyamatrix}. Due to the complexity (in particular, the dependency of message covers) brought up by messages relay, we divide the universe into six cases to consider.

\begin{theorem}
\label{t}
\emph {Algorithm 1} satisfies the validity and the convergence conditions.
\end{theorem}

From the code of \emph {Algorithm 1}, we know that $v_i[t]=a_{i} v_i[t-1]+\sum_{m\in \calM_{i}^*[t]} a_{i} \, w_m$, where $a_i=\frac{1}{|\calM^*_i[t]|+1}$. Theorem \ref{claim_1} says that we can rewrite $a_{i} v_i[t-1]+\sum_{m\in \calM_{i}^*[t]} a_{i} \, w_m$ as
\[
\sum_{j\in \calV-\calF} {\bf M}_{ij}[t]v_j[t-1],
\]
where ${\bf M}_{ij}[t]$s together satisfy the preceding four conditions. By ``stacking'' (\ref{e_matrix_i}) for different
$i$, $1\leq i\leq n-\phi$, we can
represent the state update for all the fault-free nodes together
using (\ref{e_matrix})
below, where ${\bf{M}}[t]$ is a $(n-\phi)\times (n-\phi)$ row stochastic matrix, with its $i$-th row
being equal to ${\bf{M}}_i[t]$ in (\ref{e_matrix_i}).
\begin{eqnarray}
{\bf{v}}[t] & = & {\bf{M}}[t] ~ {\bf{v}}[t-1].
\label{e_matrix}
\end{eqnarray}
By repeated application of (\ref{e_matrix}), we obtain:
\begin{eqnarray*}
{\bf{v}}[t] & = & \left(\,\Pi_{\tau=1}^t {\bf{M}}[\tau]\,\right)\, {\bf{v}}[0].
\end{eqnarray*}
As the backward product $\Pi_{\tau=1}^t {\bf{M}}[\tau]$ is a row-stochastic matrix, it holds that $\mu[0]\le v_i[t]\le U[0]$ for all $i=1, \ldots, n-\phi$ and all $t$. Thus Algorithm 1 satisfies validity condition.

The convergence of $v_i[t]$ depends on the convergence of the backward product $\Pi_{\tau=1}^t {\bf{M}}[\tau]$. As a result of this, our convergence proof uses toolkit of weak-ergodic theory that is also adopted in prior work
(e.g., \cite{Jadbabaie2003,Benezit,vaidyaII,leblanc_HiCoNs}),
with some similarities to the arguments used in \cite{vaidyaII,leblanc_HiCoNs}. The last condition in Theorem \ref{claim_1} plays an important role in the proof. For completeness, we present the formal proof of Theorem \ref{t} in Appendix \ref{app:correctness}.

\section{Connection with existing work under unbounded path length}
In this section, we show that Condition NC is equivalent to the existing results on both undirected graphs and directed graphs.  

\subsection{Undirected graph under unbounded path length}
\label{sec: extension}
If $G$ is undirected, it has been shown in \cite{impossible_proof_lynch}, that $n\ge 3f+1$ and node-connectivity $2f+1$ are both necessary and sufficient for achieving Byzantine approximate consensus. We will show that when $l\ge l^*$, our Condition NC is equivalent to the above conditions.

\begin{theorem}
\label{equiUndirected}
  When $l\ge l^*$, if $G$ undirected, then $n\ge 3f+1$ and the   node-connectivity of $G$ is at least $2f+1$ if and only if $G$ satisfies Condition NC.
\end{theorem}
Informally, if the node-connectivity of $G$, denoted by $\kappa(G)$, is at most $2f$, then we are able to show that there exists a node partition $L, R, C, F$, where $L, R$ are both nonempty and $|F|\le f$, such that neither $L\cup C\Rightarrow_{l^*} R$ nor $R\cup C\Rightarrow_{l^*} L$ holds. Conversely, if $n\ge 3f+1$ and $\kappa(G)\ge 2f+1$, using Expansion Lemma we are able to show Condition NC holds. Formal proof is given in Appendix \ref{app:connection}. 

\subsection{Directed graph under unbounded path length}

Synchronous exact Byzantine consensus is considered in \cite{Tseng2014}.  
\begin{definition}[\cite{Tseng2014}]
Given disjoint subsets $A, B$, where $B$ is non-empty: \\
(i) We say $A\to B$ if and only if set $A$ contains at least $f+1$ distinct incoming neighbors of $B$. That is, $\left | \{i | ~(i,j)\in \calE, i\in A, j\in B\} \right |>f$.\\
(ii) We say $A\not\to B$ iff $A\to B$ is not true.
\end{definition}
A tight condition (both necessary and sufficient) over the graph structure is found in \cite{Tseng2014}.

\begin{theorem}[\cite{Tseng2014}]
Given a graph $G$, exact Byzantine consensus is solvable if and only if for any partition $L, C, R, F$ of $\calV(G)$, such that both $L$ and $R$ are non-empty, and $|F|\le f$, either $L\cup C\to R$, or $R\cup C\to L$.
\end{theorem}

We term this condition as Condition 1. Note that in order for $A\to B$ to hold, we only require that there are at least $f+1$ incoming neighbors of set $B$ in set $A$. It is possible that each node in $B$ has at most $f$ incoming neighbors in $A$. As a result of this observation, our Condition NC with $l=1$ is strictly stronger than Condition 1. However, it can be shown that our Condition NC with $l\ge l^*$ is equivalent to Condition 1. 

\begin{theorem}
\label{equivaDirected}
Condition NC is equivalent to Condition 1 when $l\ge l^*$.
\end{theorem}

An alternative condition is shown in \cite{Tseng2014} to be equivalent to Condition 1. We use this condition as a bridging to show the equivalence of Condition 1 and Condition NC.

\section{Discussion and Conclusion}
\label{sec:conclusion}

Throughout this paper, we assume that faulty nodes are only able to tamper message values, leaving
message paths unchanged. However, even when faulty nodes are able to tamper message paths or even fake and transmit non-existing messages,
as long as (i) the number of faked messages is finite (each faulty node $k\in \calF$ cannot create too many non-existing messages);
and (ii) for each message $m$ tampered/faked by the faulty node $k$,  $\mathsf{path}(m)$ must satisfy $k\in \calV(\mathsf{path}(m))$, i.e., the faulty node $k$ cannot conceal itself from the message path,
using the same line of arguments as in Section \ref{sec:necessary} and Section \ref{sec:sufficiency}, it can be shown that the Condition NC is also necessary and sufficient for the existence of approximate consensus under the relaxed model.

In this paper, we unify two streams of work by assuming that each node knows the topology of up to its $l$--th neighborhood and can send message to nodes that are up to $l$ hops away, where $l\ge 1$. We prove a family of necessary and sufficient conditions for the existence
of {\em iterative}\, algorithms that achieve {\em approximate Byzantine consensus}
in arbitrary directed graphs.
The class of iterative algorithms considered in this paper ensures
that, after each iteration of the algorithm, the state of each fault-free node remains
in the {\em convex hull} of the states of the fault-free nodes at the end of
the previous iteration. 

\bibliographystyle{plain}
\bibliography{IBA}

\begin{thebibliography}{10}

\bibitem{Jadbabaie2003}
Jadbabaie Ali, Lin Jie, and A.~Stephen Morse.
\newblock Coordination of groups of mobile autonomous agents using nearest
  neighbor rules.
\newblock {\em Automatic Control, IEEE Transactions on}, 48(6):988--1001, June
  2003.

\bibitem{Ben-Or1983}
Michael Ben-Or.
\newblock Another advantage of free choice (extended abstract): Completely
  asynchronous agreement protocols.
\newblock In {\em Proceedings of the Second Annual ACM Symposium on Principles
  of Distributed Computing}, PODC '83, pages 27--30, New York, NY, USA, 1983.
  ACM.

\bibitem{Ben-Or2010}
Michael Ben{-}Or, Danny Dolev, and Ezra~N. Hoch.
\newblock Simple gradecast based algorithms.
\newblock {\em CoRR}, abs/1007.1049, 2010.

\bibitem{Benezit}
Florence BŽnŽzit, Vincent Blondel, Patrick Thiran, John Tsitsiklis, and Martin
  Vetterli.
\newblock Weighted gossip: Distributed averaging using non-doubly stochastic
  matrices.
\newblock In {\em Information Theory Proceedings (ISIT), 2010 IEEE
  International Symposium on}, pages 1753--1757, June 2010.

\bibitem{AA_Fekete_aoptimal}
Alan~D. Fekete.
\newblock Asymptotically optimal algorithms for approximate agreement.
\newblock In {\em Proceedings of the fifth annual ACM symposium on Principles
  of distributed computing}, PODC '86, pages 73--87, New York, NY, USA, 1986.
  ACM.

\bibitem{impossible_proof_lynch}
Michael~J. Fischer, Nancy~A. Lynch, and Michael Merritt.
\newblock Easy impossibility proofs for distributed consensus problems.
\newblock In {\em Proceedings of the fourth annual ACM symposium on Principles
  of distributed computing}, PODC '85, pages 59--70, New York, NY, USA, 1985.
  ACM.

\bibitem{Hajnal58}
J.~Hajnal and M.S. Bartlett.
\newblock Weak ergodicity in non-homogeneous markov chains.
\newblock In {\em Mathematical Proceedings of the Cambridge Philosophical
  Society}, volume~54, pages 233--246. Cambridge Univ Press, 1958.

\bibitem{Lamport1982}
Leslie Lamport, Robert Shostak, and Marshall Pease.
\newblock The byzantine generals problem.
\newblock {\em ACM Trans. Program. Lang. Syst.}, 4(3):382--401, July 1982.

\bibitem{leblanc_HiCoNs}
Heath~J. LeBlanc, Haotian Zhang, Shreyas Sundaram, and Xenofon Koutsoukos.
\newblock Consensus of multi-agent networks in the presence of adversaries
  using only local information.
\newblock In {\em Proceedings of the 1st International Conference on High
  Confidence Networked Systems}, HiCoNS '12, pages 1--10, New York, NY, USA,
  2012. ACM.

\bibitem{Rabin1983}
Michael~O. Rabin.
\newblock Randomized byzantine generals.
\newblock In {\em Foundations of Computer Science, 1983., 24th Annual Symposium
  on}, pages 403--409, Nov 1983.

\bibitem{Tseng2014}
Lewis Tseng and Nitin Vaidya.
\newblock Iterative approximate consensus in the presence of byzantine link
  failures.
\newblock In Guevara Noubir and Michel Raynal, editors, {\em Networked
  Systems}, Lecture Notes in Computer Science, pages 84--98. Springer
  International Publishing, 2014.

\bibitem{Vaidyamatrix}
Nitin~H. Vaidya.
\newblock Matrix representation of iterative approximate byzantine consensus in
  directed graphs.
\newblock {\em CoRR}, abs/1203.1888, 2012.

\bibitem{vaidyaII}
Nitin~H. Vaidya, Christoforos~N. Hadjicostis, and Alejandro~D.
  Dom{\'{\i}}nguez{-}Garc{\'{\i}}a.
\newblock Distributed algorithms for consensus and coordination in the presence
  of packet-dropping communication links - part {II:} coefficients of
  ergodicity analysis approach.
\newblock {\em arXiv}, arXiv:1109.6392, 2011.

\bibitem{Vaidya2012IABC}
Nitin~H. Vaidya, Lewis Tseng, and Guanfeng Liang.
\newblock Iterative approximate byzantine consensus in arbitrary directed
  graphs.
\newblock In {\em Proceedings of the 2012 ACM Symposium on Principles of
  Distributed Computing}, PODC '12, pages 365--374, New York, NY, USA, 2012.
  ACM.

\bibitem{Wolfowitz}
Jacob Wolfowitz.
\newblock Products of indecomposable, aperiodic, stochastic matrices.
\newblock {\em Proceedings of the American Mathematical Society}, 14(5):pp.
  733--737, 1963.

\bibitem{IBA_broadcast_Sundaram}
Haotian Zhang and Shreyas Sundaram.
\newblock Robustness of information diffusion algorithms to locally bounded
  adversaries.
\newblock In {\em American Control Conference (ACC), 2012}, pages 5855--5861,
  June 2012.

\end{thebibliography}

\newpage
\appendix

\setlength {\parskip}{6pt}

\centerline{\Large\bf Appendices}
\section{Necessity of Condition NC}\label{app:necessary}

\begin{proof}[Proof of Theorem \ref{thm:nc}]
Theorem \ref{thm:nc} states that if a correct IABC algorithm exists for $G$,  then $G$ satisfies:
For any node partition $L, C, R, F$ of $G$ such that $L\not=\O, R\not=\O$ and $|F|\le f$, in the induced subgraph $G_F$,
at least one of the two conditions below must be true: (i) $R\cup C\Rightarrow_l L$; (ii) $L\cup C\Rightarrow_l R$.

We prove this theorem by contradiction. Let us assume that a correct IABC exists, and there exists a partition $L, C, R, F$ of $\calV(G)$ such that $L\not=\O, R\not=\O$ and $|F|\le f$, but neither $R\cup C\Rightarrow_l L$ nor $L\cup C\Rightarrow_l R$ holds, i.e.,  $R\cup C\nRightarrow_l L$ and $L\cup C\nRightarrow_l R$. 
Consider the case when all nodes in $F$, if $F\not=\O$, are faulty, and the other nodes in sets $L, C, R$ are fault-free. Note that the fault-free nodes are not aware of the identities of the faulty nodes. In addition, assume (i) each node in $L$ has initial input $\mu$, (ii) each node in $R$ has initial input $U$, such that $U>\mu+\epsilon$ for some given constant $\epsilon$, and (iii) each node in $C$, if $C\not=\O$, has initial input in the interval $[\mu, U]$.

In the \textit{Transmit step} of iteration one, suppose that each faulty node $k\in F$ sends $w=\mu^{-}<\mu$ to nodes in $N_k^{l+}\cap L$, sends $w=U^{+}>U$ to nodes in $N_k^{l+}\cap R$, and sends some arbitrary value in the interval $[\mu, U]$ to nodes in $N_k^{l+}\cap C$. For message $m$ such that the faulty node $k$ is in its transmission path, i.e.,
$k\in \calV(\mathsf{path}(m))$, if $\mathsf{destination}(m)\in L$, node $k$ resets $\mathsf{value}(m)=\mu^-$; if $\mathsf{destination}(m)\in R$, node $k$ resets $\mathsf{value}(m)=U^+$; if
$\mathsf{destination}(m)\in C$, node $k$ resets $\mathsf{value}(m)$ to be some arbitrary value in $[\mu, U]$.

Consider any node $i\in L$.
Since $|F|\le f$, we know $|N_i^{l-}\cap F|\le f$.
In addition,  $C\cup R\nRightarrow_l L$ holds in $G_F$ implies $\kappa_l(C\cup R, i)\le f$. Let $S_l$ be a minimum restricted $(C\cup R, i)$--cut in $G_F$.
From the perspective of node $i$, there exist two possible cases:

\begin{enumerate}[label=\emph{(\alph*)}]
\item Both $S_l$ and $N_i^{l-}\cap F$ are non-empty: We know $|N_i^{l-}\cap F|\le f$ and $|S_l|\le f$.  From node $i$'s perspective, two scenarios are possible: (1) nodes in $N_i^{l-}\cap F$ are faulty, all the messages relayed via them are tampered and the other nodes are fault-free, and (2) nodes in $S_l$ are faulty and the other nodes are fault-free.

In scenario (1), from node $i$'s perspective, the untampered values are in the interval $[\mu, U]$. By validity condition, $v_i[1]\ge \mu$. On the other hand, in scenario (2), the untampered values are $\mu^{-}$ and $\mu$, where $\mu^{-}<\mu$; so $v_i[1]\le \mu$, according to validity condition. Since node $i$ does not know whether the correct scenario is (1) or (2), it must update its state to satisfy the validity condition in both cases. Thus, it follows that $v_i[1]=\mu$.

\item At most one of $S_l$ and $N_i^{l-}\cap F$ is non-empty: Thus, $|S_l\cup (N_i^{l-}\cap F)|\le f$. From node $i$'s perspective, it is possible that the nodes in $S_l\cup (N_i^{l-}\cap F)$ are all faulty, the messages relayed via nodes in $S_l\cup (N_i^{l-}\cap F)$ are tampered while the rest of the nodes are fault-free. In this situation, the untampered values received by node $i$ (which are all from nodes in $N_i^{l-}\cap L$) are all $\mu$, and therefore, $v_i[1]$ must be set to $\mu$ as per the validity condition.
\end{enumerate}

At the end of iteration 1: for each node $i$ in $L$ $v_i[1]=\mu$; similarly, for each node $j$ in $R$, $v_j[1]=U$; if $C\not=\O$, for each node $i$ in $C$, $v_i[1]\in [\mu, U]$. All these conditions are identical to the condition when $t=0$. Then by a repeated application of of above argument, it follows that for any $t\ge 0$, $v_i[t]=\mu$ for all $i\in L$, $v_j[t]=U$ for all $j\in R$ and $v_k[t]\in [\mu, U]$ for all $k\in C$, if $C\not=\O$.

Since $L$ and $R$ both contain fault-free nodes, the convergence requirement is not satisfied. This contradicts the assumption that a correct iterative algorithm exists.
\end{proof}

\subsection{Lower bound on graph size and nodes' incoming degrees} \label{app:cdegree}

\begin{proof}[Proof of Corollary \ref{cdegree}]
Corollary \ref{cdegree} states that if $G$ satisfies Condition NC, then
 $n$ must be at least $3f+1$, and each node must have at least $2f+1$ incoming neighbors other than itself, i.e., $|N_i^{-}-\{i\}|\ge 2f+1$.

The main
techniques used in this proof are fairly routine, and are given here
largely for both concreteness and completeness.

We first show the claim that $n\ge 3f+1$.
For $f=0$, $n\ge 3f+1=1$ is trivially true. For $f>0$, the proof is by contradiction. Suppose that $2\le n\le 3f$. In this case, we can partition $\calV(G)$ into sets $L, R, C, F$ such that $1\le |L|\le f$, $1\le |R|\le f$, $0\le |F|\le f$ and $|C|=0$, i.e., $C$ is empty. Since $1\le |L\cup C|=|L|\le f$ and $1\le |R\cup C|=|R|\le f$, we have $L\cup C\not\Rightarrow_l R$ and $R\cup C\not\Rightarrow_l L$, respectively in $G_F$. This contradicts the assumption that $G$ satisfies Condition NC. Thus, $n\ge 3f+1$.

It remains to show $|N_i^{-}-\{i\}|\ge 2f+1$.
%
%
Suppose that, contrary to our claim, there exists a node $i$ such that $|N_i^{-}-\{i\}|\le 2f$.
Define set $L=\{i\}$ and partition $N_i^{-}-\{i\}$ into two sets $F$ and $H$
such that $|H|=\lfloor |N_i^{-}-\{i\}|/2\rfloor\leq f$
and $|F|=\lceil |N_i^{-}-\{i\}|/2\rceil\leq f$. Note that $H=\O, F=\O$ if and only if $f=0$.
Define $R=\calV(G)-F-L=\calV(G)-F-\{i\}$ and $C=\O$. Since $|\calV(G)| = n \geq \max(2,3f+1)$, $R$ is non-empty. From the construction of $R$, we have $N_i^-\cap R=H$, and $|N_i^-\cap R|=|H|\leq f$.
Since $L=\{i\}$, $|N_i^-\cap R|\leq f$ and $C=\O$, it follows that $R\cup C\not\Rightarrow_l L$.
On the other hand, as $|L|=1<f+1$, we have $L\cup C\not\Rightarrow_l R$. This violates the assumption that $G$ satisfies Condition NC. The proof is complete.
\end{proof}

\subsection{Lower bound on graph density}
\begin{proof}[Proof of Proposition \ref{density}]
Proposition \ref{density} states that: For $f=1$ and $l_0=1$, there exists a graph $G$ for any $n\ge 3f+1=4$ such that
(i) $|N_i^{-}|=4$ for all $i\in \calV(G)$; and
(ii) $|\calE(G)|=4n$.

We prove this proposition by inducting on $n$.
In the complete graph with $n=4$, $|N_i^{-}|=4$ ( including $i$ itself ) for all $i\in \calV(G)$ and the total number of edges is $16$. So the base case easily follows. Assume that the proposition holds for $n>4$. Let $G$ be a graph with $|\calV(G)|=n$,  $|N_i^{-}|=4$ for all $i\in \calV(G)$ and $|\calE(G)|=4n$. Let $x\notin \calV(G)$, add self-loop to $x$ and connect arbitrary $3$ nodes in $G$ to node $x$. Denote the resulting graph as $G^{\prime}$. Note that the only outgoing edge of $x$ is its self-loop.  Let $L, R, C$ and $F$ be an arbitrary node partition of $G^{\prime}$ such that $L, R$ are nonempty and $|F|\le 1$.

For the case when $L=\{x\}$, since $N_x^{-}=4$ and $|F|\le 1$, we know $R\cup C\Rightarrow_1 L$. Similarly we can show the case when $R=\{x\}$. When $L\not=\{x\}$ and $R\not=\{x\}$, let $L^{\prime}=L-\{x\}$, $C^{\prime}=C-\{x\}$, $R^{\prime}=R-\{x\}$ and $F^{\prime}=F-\{x\}$, then the obtained $L^{\prime}, R^{\prime}, C^{\prime}$ and $F^{\prime}$ is a node partition of the original graph $G$ such that $L^{\prime}, R^{\prime}$ are nonempty and $|F^{\prime}|\le 1$. Since $G$ satisfies Condition NC, then either $L^{\prime}\cup C^{\prime}\Rightarrow_1 R^{\prime}$ or $R^{\prime}\cup C^{\prime}\Rightarrow_1 L^{\prime}$. As $G^{\prime}$ inherits every edge in $G$, we have either $L\cup C\Rightarrow_1 R$ or $R\cup C\Rightarrow_1 L$ in $G^{\prime}$. This completes the induction.
\end{proof}

\subsection{Equivalence of Condition NC and single source component condition}\label{app:equi}

\begin{proof}[Proof of Theorem \ref{thm:nc2}]
Theorem \ref{thm:nc2} states that graph $G$ satisfies Condition NC if and only if 
every reduced graph $\widetilde{G^l}_F$ obtained as per Definition \ref{def:reduced}
must contain exactly one {\em source component}.

We first show that if graph $G$ satisfies Condition NC, then every reduced graph of $G^l$ contains exactly one source component.

 For any reduced graph $\widetilde{G^l}_F$, the meta-graph $(\widetilde{G^l}_F)^{SCC}$ is a DAG and finite. Thus, at least one source component must exist in $\widetilde{G^l}_F$.
 We now prove that $\widetilde{G^l}_F$ cannot
contain more than one source component. The proof is by contradiction.
Suppose that there exists a set $F\subseteq \calV(G)$ with $|F|\leq f$, and a reduced graph
$\widetilde{G^l}_F$ corresponding to $F$, such
that $\widetilde{G^l}_F$ contains at least two source components, say $K_1$ and $K_2$, respectively.
Let $L=K_1$, $R=K_2$, and $C=\calV(G)-F-L-R$. Then $L, R, C$ together with the given $F$ form a node partition of $\calV(G)$ such that $L\not=\O, R\not=\O$ and $|F|\le f$.

Since graph $G$ satisfies Condition NC, without loss of generality, assume that $R\cup C \Rightarrow_l L$, i.e., there exists a node $i\in L$ such that $\kappa_l(R\cup C, i)\ge f+1$ in $G_F$.
On the other hand,
since $L$ is a source component in $\widetilde{G^l}_F$, by the definition of reduced graph, we know all paths from $R\cup C$ to node $i$ of length at most $l$ in $G$ are covered by $C_i\cup F$, where $C_i$ is defined preceding Definition \ref{def:reduced}.
Thus, $C_i$ is a restricted $(R\cup C, i)$--cut of $G_F$. However, by construction of $\widetilde{G^l}_F$, the size of $C_i$ is at most $f$. So we arrive at a contradiction.

To complete the equivalence proof it remains to show that if every reduced graph contains exactly one source component, then the graph must satisfy Condition NC.

Suppose, on the contrary, that $G$ does not satisfy Condition NC. Then there exists a node partition $L, R, C$ and $F$ of $G$ with $L, R$ are nonempty and $|F|\le f$ such that $L\cup C\not\Rightarrow_l R$ and $R\cup C\not\Rightarrow_l L$ in $G_F$. By the definition of the relation $\not\Rightarrow_l$, there is no path of length at most $l$ from $L\cup C$ to a node in $R$, and no path of length at most $l$ from $R\cup C$ to a node in $L$. This further implies that no nodes in $R\cup C$ can reach a node in $L$ in $\widetilde{G^l}_F$ and no nodes in $L\cup C$ can reach a node in $R$ in $\widetilde{G^l}_F$. Thus both $L$ and $R$ are source components, contradicting the condition that there is only one source component in every $\widetilde{G^l}_F$.

\end{proof}

%
%
%

\section{Sufficiency of Condition NC}\label{sec:sufficiencyProof}

\subsection{The trimming function is well-defined}
\label{trimmingProof}
\begin{proof}[Proof of Theorem \ref{trimming}]
Theorem \ref{trimming} states that if graph $G$ satisfies Condition NC, then the sets of messages $\calM_{is}[t]$, $\calM_{il}[t]$ are well-defined and  $\calM_{i}^{*}[t]$ is nonempty.

For ease of exposition, with a slight abuse of notation, we drop the time indices of $\calM_{i}^{\prime}[t]$, $\calM_{is}[t]$, $\calM_{il}[t]$ and $\calM_{i}^{*}[t]$, respectively.
From Corollary \ref{cdegree}, we know $|N_i^--\{i\}|\ge 2f+1$. Since $|\calT^*(\calM_{is})|=f$ and $|\calT^*(\calM_{il})|=f$, the message from at least one incoming neighbor of node $i$ is not covered by $\calT^*(\calM_{is})\cup \calT^*(\calM_{il})$. So $\calM_{i}^{*}$ is nonempty. 

We prove the existence of $\calM_{is}$ and $\calM_{il}$ by construction. The set $\calM_{is}$ can be constructed using the following algorithm, which can be easily adapted for the construction of set $\calM_{il}$. For clarity of proof, we construct $\calM_{is}$ and $\calM_{il}$ sequentially, although they can be found in parallel.

As before, sort the messages in $\calM_i^{\prime}$ in an increasing order according to their messages values. Initialize $\calM_{is}\gets \O, \, Q\gets\O$ and $ \calM \gets \calM_{i}^{\prime}$. At each round, let $m_s$ be a message with the smallest value in $\calM$, and update $Q$, $\calM$ as follows,
\begin{align*}
&Q\gets Q\cup \{m_s\};\\
&\calM\gets \calM-\{m_s\}.
\end{align*}
If $|\calT^*(Q)|\ge f+1$, set $\calM_{is}\gets Q-m_s$ and return $\calM_{is}$; otherwise, repeat this procedure.

If the algorithm terminates, then by the code, it is easy to see that the returned $\calM_{is}$ satisfies the following conditions: For all $m\in \calM_i^{\prime}-\calM_{is}$ and $m^{\prime}\in \calM_{is}$ we have $\mathsf{value}(m)\ge \mathsf{value}(m^{\prime})$;
and the cardinality of a minimum cover of $\calM_{is}$ is exactly $f$, i.e., $|\calT^*(\calM_{is})|=f$.
It remains to show this algorithm terminates. Suppose this algorithm does not terminate.
The problem of finding a minimum cover of a set of messages, i.e., computing $\calT^*(Q)$, can be converted to the problem of finding a minimum cut of a vertex pair, by adding a new vertex $y$ and connecting $y$ to every vertex in $\calV(G)-\{i\}$. The latter problem can be solved in polynomial time. Thus, non-termination implies that $|\calT^*(\calM_i^{\prime})|\le f$, which further implies that
the $l$--restricted $(\calV(G)-\{i\}, i)$--connectivity is less than or equal to $f$. On the other hand, consider the node partition that $L=\{i\}$, $R=\calV(G)-\{i\}$, and $C=F=\O$, neither $L\cup C\Rightarrow_l R$ nor $R\cup C\Rightarrow_l L$ holds. This contradicts the assumption that $G$ satisfies Condition NC. So the above algorithm terminates.

We can adapt the above procedure to construct $\calM_{il}$ by modifying the initialization step to be $Q\gets \O$,  $\calM\gets \calM_i^{\prime}-\calM_{is}$.
Termination can be shown similarly. Suppose this algorithm does not terminate. Non-termination implies that $|\calT^*(\calM_i^{\prime}-\calM_{is})|\le f$, which further implies that in the node partition $L=\{i\}$, $F=\calT^*(\calM_{is})$, $R=\calV(G)-F-L$, $C=\O$, the $l$--restricted $(R\cup C, \{i\})$--connectivity is no more than $f$, i.e., $R\cup C\nRightarrow_l L$. In addition, since $|L|=1$, $L\cup C\nRightarrow_l R$. This contradicts the assumption that $G$ satisfies Condition NC.
Therefore, $\calM_{is}$ and $\calM_{il}$ are well-defined.

\end{proof}

\subsection{Matrix Preliminaries}\label{MatrixPreliminaries}


For a row stochastic matrix $\bf{A}$,
 coefficients of ergodicity $\delta(\bf{A})$ and $\lambda(\bf{A})$ are defined as
\cite{Wolfowitz}:
\begin{align}
\delta({\bf{A}}) & :=   \max_j ~ \max_{i_1,i_2}~ | {\bf{A}}_{i_1\,j}-{\bf{A}}_{i_2\,j} |, \label{e_delta} \\
\lambda({\bf{A}}) & :=  1 - \min_{i_1,i_2} \sum_j \min({\bf{A}}_{i_1\,j} ~, {\bf{A}}_{i_2\,j}). \label{e_lambda}
\end{align}
It  is easy to see that  $0\leq \delta({\bf{A}}) \leq 1$ and $0\leq \lambda({\bf{A}}) \leq 1$, and that the rows are all identical if and and only if $\delta(\textbf{A})=0$. Additionally, $\lambda({\bf{A}}) =0$ if and only if $\delta({\bf{A}}) = 0$.

The next result from \cite{Hajnal58} establishes a relation between the coefficient of ergodicity $\delta(\cdot)$ of a product of row stochastic matrices, and the coefficients of ergodicity $\lambda(\cdot)$ of the individual matrices defining the product.

\begin{claim}
\label{claim_delta}
For any $p$ square row stochastic matrices ${\bf{Q}}(1),{\bf{Q}}(2),\dots ,{\bf{Q}}(p)$,
\begin{align}
\delta({\bf{Q}}(1){\bf{Q}}(2)\cdots {\bf{Q}}(p)) ~\leq ~
 \Pi_{i=1}^p ~ \lambda({\bf{Q}}(i)).
\end{align}
\end{claim}
Claim \ref{claim_delta} is proved in \cite{Hajnal58}. It implies that
if, for all $i$, $\lambda({\bf{Q}}(i))\leq 1-\gamma$ for some $\gamma>0$, then $\delta({\bf{Q}}(1){\bf{Q}}(2)\cdots {\bf{Q}}(p))$ will approach zero as $p$ approaches $\infty$.

\begin{definition}
A row stochastic
 matrix ${\bf{H}}$ is said to be a {\em scrambling}\, matrix, if $\lambda({\bf{H}})<1$
{\normalfont \cite{Hajnal58,Wolfowitz}}.
\end{definition}

In a scrambling matrix ${\bf{H}}$, since $\lambda({\bf{H}})<1$, for each pair of
rows $i_1$ and $i_2$, there exists a column $j$ (which may depend on
$i_1$ and $i_2$) such that
 ${\bf{H}}_{i_1\,j}>0$ and ${\bf{H}}_{i_2\,j}>0$, and vice-versa \cite{Hajnal58,Wolfowitz}.
As a special case, if any one column of a row stochastic matrix ${\bf{H}}$
contains only non-zero entries that are lower bounded by some
constant $\gamma>0$, then ${\bf{H}}$ must be scrambling, and $\lambda({\bf{H}})\leq 1-\gamma$.

\begin{definition}
For matrices ${\bf{A}}$ and ${\bf{B}}$ of identical size, and
a scalar $\gamma$, ${\bf{A}}\leq \gamma \, {\bf{B}}$ provided
that ${\bf{A}}_{ij}\leq \gamma\, {\bf{B}}_{ij}$ for all $i,j$.
\end{definition}

\subsection{Matrix representation}\label{app:claim_1}
Some relevant corollaries and concepts are needed before we are able to proceed to the proof of Theorem \ref{claim_1}.  

\begin{corollary}
\label{claim_suff}
Suppose that graph $G$ satisfies Condition NC. Then it follows that in each reduced graph $\widetilde{G^l}_F \in R_F$,
there exists at least one node that has directed paths to all the nodes in $\widetilde{G^l}_F$.
\end{corollary}
Corollary \ref{claim_suff} follows immediately from Theorem \ref{thm:nc2}.
\begin{corollary}
\label{l_one_column}
Suppose that $G$ satisfies Condition NC. Let $|F|=\phi$, for any $\widetilde{G^l}_F\in R_F$ with ${\bf{H}}$ as the adjacency matrix, ${\bf{H}}^{n-\phi}$ has at least one non-zero column.
\end{corollary}

\begin{proof}
By Corollary \ref{claim_suff}, in graph $\widetilde{G^l}_F$ there exists at least one node, say
node $k$, that has a directed path in $\widetilde{G^l}_F$ to all the remaining nodes in $\calV_F$, i.e., $\calV(G)-F$.
Since the length of the path from $k$ to any other node in $\widetilde{G^l}_F$ can contain
at most $n-\phi-1$ directed edges,
the $k$-th column of matrix ${{\bf{H}}^{n-\phi}}$ will
be non-zero.\footnote{That is, all the entries of the column will be
non-zero (more precisely, positive, since the entries of matrix $\bf{H}$
are non-negative).
Also, such a non-zero column will exist in ${\bf H}^{n-\phi-1}$ too.
We use the loose bound of $n-\phi$ to simplify the presentation. }
\end{proof}

\begin{definition}
We will say that an entry of a matrix is ``non-trivial'' if it is lower
bounded by $\beta$, where $\beta$ is some constant to be defined later.
\end{definition}




\begin{proof}[Proof of Theorem \ref{claim_1}]
 
Recall that nodes 1 through $n-\phi$ are fault-free, and the remaining
$\phi$ nodes ($\phi\leq f$) are faulty.
Consider a fault-free node $i$ performing the {\em update step}
in \emph {Algorithm 1}.
Recall that $\calM_{is}[t]$ and $\calM_{il}[t]$ messages are eliminated from $\calM_i[t]$. Let $\calS_{ig}[t]\subseteq \calM_{is}[t]$ and $\calL_{ig}[t]\subseteq \calM_{il}[t]$, respectively, be the sets of removed messages that are not covered by faulty nodes.
%
Let $\calP_i^{*}[t]$ be the set of paths corresponding to all the messages in $\calM^{*}_i[t]$.
\emph{Untampered message representation} of the evolution of $v_i$ and construction of ${\bf{M}}_i[t]$ differ somewhat depending on whether sets $\calL_{ig}[t], \calS_{ig}[t]$ and $\calP_i^{*}[t]\cap \calF$ are empty or not, where $\calP_i^{*}[t]\cap \calF=\O$ means that no message in $\calM_i^*[t]$ has been tampered by faulty nodes and $\calP_i^{*}[t]\cap \calF\not=\O$ means that there exists a message that is tampered by faulty nodes. It is possible that $\calT^*(\calM_{is}[t])=\calT^*(\calM_{il}[t])=\calF$, which means all messages in $\calM_{is}[t]$ and $\calM_{il}[t]$ are tampered by faulty nodes, i.e., $\calS_{ig}[t]=\O$ and $\calL_{ig}[t]=\O$. We divide the possibilities into six cases:

\begin{enumerate}
\item Case I: $\calS_{ig}[t]\not=\O, \calL_{ig}[t]\not=\O$ and $\calP_i^{*}[t]\cap \calF\not=\O$.
\item Case II: $\calS_{ig}[t]\not=\O, \calL_{ig}[t]\not=\O$ and $\calP_i^{*}[t]\cap \calF=\O$.
\item Case III: one of $\calS_{ig}[t], \calL_{ig}[t]$ is empty and $\calP_i^{*}[t]\cap \calF\not=\O$.
\item Case IV: one of $\calS_{ig}[t], \calL_{ig}[t]$ is empty and $\calP_i^{*}[t]\cap \calF=\O$.
\item Case V: $\calS_{ig}[t]=\O, \calL_{ig}[t]=\O$ and $\calP_i^{*}[t]\cap \calF\not=\O$.
\item Case VI: $\calS_{ig}[t]=\O, \calL_{ig}[t]=\O$ and $\calP_i^{*}[t]\cap \calF=\O$.
\end{enumerate}

We first describe the construction of ${\bf M}_i[t]$ in case I, when $\calS_{ig}[t]\not=\O, \calL_{ig}[t]\not=\O$ and $\calP_i^{*}[t]\cap \calF\not=\O$.
Let $\bar{w}_{is}[t]$ and $\bar{w}_{il}[t]$ be defined as shown below. Recall that $w_m=\mathsf{value}(m)$.


\begin{eqnarray}
\bar{w}_{is}[t]=\frac{\sum_{m\in \calS_{ig}[t]}w_{m}}{|\calS_{ig}[t]|} \quad ~ \text{and}~\quad
\bar{w}_{il}[t]=\frac{\sum_{m\in \calL_{ig}[t]}w_{m}}{|\calL_{ig}[t]|}.
\end{eqnarray}

By the definitions of $\calS_{ig}[t]$ and $\calL_{ig}[t]$, $\bar{w}_{is}\le w_{m^{\prime}}\le \bar{w}_{il}$, for each message $m^{\prime}\in \calM^*_i[t]$. Thus, for each message $m^{\prime}$, we can find convex coefficient $\gamma_{m^{\prime}}$, where $0\le \gamma_{m^{\prime}} \le 1$, such that
\begin{align*}
w_{m^{\prime}}&=\gamma_{m^{\prime}}\bar{w}_{is}+(1-\gamma_{m^{\prime}})\bar{w}_{il}\\
&=\frac{\gamma_{m^{\prime}}}{|\calS_{ig}[t]|}\sum_{m\in \calS_{ig}[t]}w_m+\frac{1-\gamma_{m^{\prime}}}{|\calL_{ig}[t]|}\sum_{m\in \calL_{ig}[t]}w_m.
\end{align*}

Recall that in \emph {Algorithm 1}, $v_i[t]=a_iv_i[t-1]+\sum_{m\in \calM_i^*[t]}a_iw_m$, where $a_i=\frac{1}{|\calM_i^*[t]|+1}$. In case I, since $\calP_i^{*}[t]\cap \calF\not=\O$, there exist messages in $\calM_i^{*}[t]$ that are tampered by faulty nodes. We need to replace these ``bad messages" by ``good messages" in the evolution of $v_i$. In particular,

\begin{align}
v_i[t]&=a_iv_i[t-1]+\sum_{m\in \calM_i^*[t]}a_iw_m\\
&=a_iv_i[t-1]+\sum_{m\in \calM_i^*[t]:~\calV(\mathsf{path}(m))\cap \calF=\O~}a_iw_m
+\sum_{m\in \calM_i^*[t]:~ \calV(\mathsf{path}(m))\cap \calF\not=\O~}a_iw_m\\
&=a_iv_i[t-1]+\sum_{m\in \calM_i^*[t]:~  \calV(\mathsf{path}(m))\cap \calF=\O~}a_iw_m\\
&\quad+\sum_{m\in \calM_i^*[t]:~\calV(\mathsf{path}(m))\cap \calF\not=\O~}a_i(\frac{\gamma_{m}}{|\calS_{ig}[t]|}\sum_{m^{\prime}\in \calS_{ig}[t]}w_{m^{\prime}}+\frac{1-\gamma_{m}}{|\calL_{ig}[t]|}\sum_{m^{\prime}\in \calL_{ig}[t]}w_{m^{\prime}})\\
&=a_iv_i[t-1]+\sum_{m\in \calM_i^*[t]:~ \calV(\mathsf{path}(m))\cap \calF=\O~}a_iw_m\\
&\quad+\sum_{m^{\prime}\in \calS_{ig}[t]}\big{(}\sum_{m\in \calM_i^*[t]:~ \calV(\mathsf{path}(m))\cap \calF\not=\O~}\frac{a_i\gamma_m}{|\calS_{ig}[t]|} \big{)}w_{m^{\prime}}\\
&\quad+\sum_{m^{\prime}\in \calL_{ig}[t]}\big{(}\sum_{m\in \calM_i^*[t]:~ \calV(\mathsf{path}(m))\cap \calF\not=\O~}\frac{a_i(1-\gamma_m)}{|\calL_{ig}[t]|} \big{)}w_{m^{\prime}}.
\end{align}
That is, $v_i[t]$ can be represented as a convex combination of values of untampered messages collected at iteration $t$, where $v_i[t-1]=\mathsf{value}(v_i[t-1], (i, i))$.
For future reference, we refer to the above convex combination as \emph{untampered message representation of $v_i[t]$} in case I and the convex coefficient of each message in the untampered message representation as \emph{message weight}.

Note that if $m$ is an untampered message in $\calM_i^*[t]$ or
$m\in \calS_{ig}[t]\cup \calL_{ig}[t]$, then $w_m=v_j[t-1]$ holds, where node $j$ is the source of message $m$, i.e., $\mathsf{source}(m)=j$. $v_i[t]$ can be further rewritten as follows, where $\mathbbm{1}\{x\}=1$ if $x$ is true, and $\mathbbm{1}\{x\}=0$, otherwise.

\begin{align*}
v_i[t]&=\sum_{j\in \calV-\calF}v_j[t-1]\Big{(} a_i\mathbbm{1}\{j=i\}+
\sum_{m\in \calM_i^*[t]:~ \calV(\mathsf{path}(m))\cap \calF=\O~}a_i\mathbbm{1}\{\mathsf{source}(m)=j\}\\
&\quad +\sum_{m^{\prime}\in \calS_{ig}[t]}\big{(}\sum_{m\in \calM_i^*[t]:~ \calV(\mathsf{path}(m))\cap \calF\not=\O~}
\frac{a_i\gamma_{m}}{|\calS_{ig}[t]|}\mathbbm{1}\{\mathsf{source}(m^{\prime})=j\}\big{)}\\
&\quad+
\sum_{m^{\prime}\in \calL_{ig}[t]}\big{(}\sum_{m\in \calM_i^*[t]:~\calV(\mathsf{path}(m))\cap \calF\not=\O~}\frac{a_i(1-\gamma_{m})}{|\calL_{ig}[t]|}\mathbbm{1}\{\mathsf{source}(m^{\prime})=j\}\big{)}\Big{)},
\end{align*}

Thus, for each node $i, j\in \calV-\calF$, define the entry ${\bf M}_{ij}[t]$ as follows,

\begin{align*}
      {\bf{M}}_{ij}[t]&=a_i\mathbbm{1}\{j=i\}+
\sum_{m\in \calM_i^*[t]:~ \calV(\mathsf{path}(m))\cap \calF=\O~}a_i\mathbbm{1}\{\mathsf{source}(m)=j\}\\
&\quad +\sum_{m^{\prime}\in \calS_{ig}[t]}\big{(}\sum_{m\in \calM_i^*[t]:~ \calV(\mathsf{path}(m))\cap \calF\not=\O~}
\frac{a_i\gamma_{m}}{|\calS_{ig}[t]|}\mathbbm{1}\{\mathsf{source}(m^{\prime})=j\}\big{)}\\
&\quad+
\sum_{m^{\prime}\in \calL_{ig}[t]}\big{(}\sum_{m\in \calM_i^*[t]:~ \calV(\mathsf{path}(m))\cap \calF\not=\O~}\frac{a_i(1-\gamma_{m})}{|\calL_{ig}[t]|}\mathbbm{1}\{\mathsf{source}(m^{\prime})=j\}\big{)}.
\end{align*}
The third condition in Theorem \ref{claim_1}
trivially follows from the above construction.
 By above definition, ${\bf M}_{ij}\ge a_i$, where ${\bf M}_{ij}> a_i$ holds when there exists a nontrivial cycle (not a self-loop) of length at most $l$ that contains node $i$ and no faulty nodes. In addition, $a_i\ge \alpha$ by (\ref{lowerbound}). Thus, ${\bf M}_{ii}[t]\ge \alpha.$
The second condition holds.
 Now we show that ${\bf{M}}_{i}[t]$ is a stochastic vector. It is easy to see that ${\bf{M}}_{ij}[t]\ge 0$. In addition, we have

\begin{align*}
      \sum_{j\in \calV-\calF}{\bf{M}}_{ij}[t]&=\sum_{j\in \calV-\calF}
      \Big{(} a_i\mathbbm{1}\{j=i\}+
\sum_{m\in \calM_i^*[t]:~\calV(\mathsf{path}(m))\cap \calF=\O~}a_i\mathbbm{1}\{\mathsf{source}(m)=j\}\\
&\quad +\sum_{m^{\prime}\in \calS_{ig}[t]}\big{(}\sum_{m\in \calM_i^*[t]:~\calV(\mathsf{path}(m))\cap \calF\not=\O~}
\frac{a_i\gamma_{m}}{|\calS_{ig}[t]|}\mathbbm{1}\{\mathsf{source}(m^{\prime})=j\}\big{)}\\
&\quad+
\sum_{m^{\prime}\in \calL_{ig}[t]}\big{(}\sum_{m\in \calM_i^*[t]:~\calV(\mathsf{path}(m))\cap \calF\not=\O~}\frac{a_i(1-\gamma_{m})}{|\calL_{ig}[t]|}\mathbbm{1}\{\mathsf{source}(m^{\prime})=j\}\big{)}\Big{)}\\
      &=a_i\sum_{j\in \calV-\calF}\mathbbm{1}\{i=j\}+\sum_{m\in \calM_i^*[t]:~\mathsf{path}(m)\cap \calF=\O~}a_i \sum_{j\in \calV-\calF}\mathbbm{1}\{ \mathsf{source}(m)=j\}\\
      &\quad+\sum_{m\in \calM_i^*[t]:~\mathsf{path}(m)\cap \calF\not=\O}\Big{(}\frac{a_i\gamma_m}{|\calS_{ig}[t]|} \sum_{ m^{\prime}\in \calS_{ig}[t]}\sum_{j\in \calV-\calF}\mathbbm{1}\{ \mathsf{source}(m^{\prime})=j\}\Big{)}\\
      &\quad + \sum_{m\in \calM_i^*[t]:~\mathsf{path}(m)\cap \calF\not=\O}\Big{(}\frac{a_i(1-\gamma_m)}{|\calL_{ig}[t]|}\sum_{m^{\prime}~\in \calL_{ig}[t]}\sum_{j\in \calV-\calF}\mathbbm{1}\{\mathsf{source}(m^{\prime})=j\}\Big{)}\\
      &=a_i+\sum_{m\in \calM_i^*[t]:~\mathsf{path}(m)\cap \calF=\O~}a_i\\
      &\quad+\sum_{m\in \calM_i^*[t]:~\mathsf{path}(m)\cap \calF\not=\O} \frac{a_i\gamma_m}{|\calS_{ig}[t]|} \sum_{m^{\prime}~\in \calS_{ig}[t]}1\\
      &\quad + \sum_{m\in \calM_i^*[t]:~\mathsf{path}(m)\cap \calF\not=\O}\frac{a_i(1-\gamma_m)}{|\calL_{ig}[t]|}\sum_{m^{\prime}~\in \calL_{ig}[t]}1\\
      &=a_i+\sum_{m\in \calM_i^*[t]:~\mathsf{path}(m)\cap \calF=\O~}a_i+\sum_{m\in \calM_i^*[t]:~\mathsf{path}(m)\cap \calF\not=\O} \frac{a_i\gamma_m}{|\calS_{ig}[t]|} |\calS_{ig}[t]|\\
      &\quad + \sum_{m\in \calM_i^*[t]:~\mathsf{path}(m)\cap \calF\not=\O}\frac{a_i(1-\gamma_m)}{|\calL_{ig}[t]|}|\calL_{ig}[t]|\\
      &=a_i+\sum_{m\in \calM_i^*[t]:~\mathsf{path}(m)\cap \calF=\O~}a_i+\sum_{m\in \calM_i^*[t]:~\mathsf{path}(m)\cap \calF\not=\O}a_i\\
      &=a_i(|\calM^*_i[t]|+1)\\
      &=1.
\end{align*}
So ${\bf{M}}_{i}[t]$ is row stochastic.

In case II, since $\calP_i^{*}[t]\cap \calF=\O$, all messages in $\calM^*_i[t]$ are untampered by faulty nodes. Let $m_0$ be an arbitrary message in $\calM^*_i[t]$, with $\mathsf{source}(m_0)=j^*$. In order to guarantee condition 4) holds, we rewrite $v_i[t]$ as follows,
\begin{align*}
v_i[t]&=a_iv_i[t-1]+\sum_{m\in \calM_i^*[t]}a_iw_m\\
      &=a_iv_i[t-1]+a_iw_{m_0}+\sum_{m\in \calM_i^*[t]-\{m_0\}}a_iw_m\\
      &=a_iv_i[t-1]+\frac{1}{2}a_iw_{m_0}+\frac{1}{2}a_iw_{m_0}+\sum_{m\in \calM_i^*[t]-\{m_0\}}a_iw_m\\
      &=a_iv_i[t-1]+\frac{1}{2}a_iw_{m_0}+\frac{1}{2}a_i(\frac{\gamma_{m_0}}{|\calS_{ig}[t]|}\sum_{m^{\prime}\in \calS_{ig}[t]}w_{m^{\prime}}+\frac{1-\gamma_{m_0}}{|\calL_{ig}[t]|}\sum_{m^{\prime}\in \calL_{ig}[t]}w_{m^{\prime}})\\
      &\quad+\sum_{m\in \calM_i^*[t]-\{m_0\}}a_iw_m\\
      &=a_iv_i[t-1]+\frac{1}{2}a_iw_{m_0}+\sum_{m^{\prime}\in \calS_{ig}[t]}\frac{a_i\gamma_{m_0}}{2|\calS_{ig}[t]|}w_{m^{\prime}}+\sum_{m^{\prime}\in \calL_{ig}[t]}\frac{a_i(1-\gamma_{m_0})}{2|\calL_{ig}[t]|}w_{m^{\prime}}\\
      &\quad+\sum_{m\in \calM_i^*[t]-\{m_0\}}a_iw_m.
\end{align*}
Note that we did not use the above trick in case I. This is because, in case I, by substituting tampered messages in $\calM_i^*[t]$ by untampered messages in $\calS_{ig}[t]$ and $\calL_{ig}[t]$, as will be seen later, condition 4) is automatically guaranteed.

We refer to the above convex combination as the \emph{untampered message representation of $v_i[t]$} in case II. And the convex coefficient of each message in the above representation as \emph{weight assigned} to that message. Combining the coefficients of messages according to message sources, it is obtained that

\begin{align*}
v_i[t]&=\sum_{j\in \calV-\calF}v_j[t-1]\Big{(}a_i\mathbbm{1}\{i=j\}
      +\frac{1}{2}a_i\mathbbm{1}\{j=j^*\}+\sum_{m\in \calM_i^*[t]-\{m_0\}}a_i\mathbbm{1}\{\mathsf{source}(m)=j\}\\
      &\quad +\frac{a_i\gamma_{m_0}}{2|\calS_{ig}[t]|}\sum_{m^{\prime}\in \calS_{ig}[t]}\mathbbm{1}\{\mathsf{source}(m^{\prime})=j\}+\frac{a_i(1-\gamma_{m_0})}{2|\calL_{ig}[t]|}\sum_{m^{\prime}\in \calL_{ig}[t]}\mathbbm{1}\{\mathsf{source}(m^{\prime})=j\}
      \Big{)}.
\end{align*}
Thus, define ${\bf M}_{ij}$ by
\begin{align*}
{\bf M}_{ij}&=a_i\mathbbm{1}\{i=j\}
      +\frac{1}{2}a_i\mathbbm{1}\{j=j^*\}+\sum_{m\in \calM_i^*[t]-\{m_0\}}a_i\mathbbm{1}\{\mathsf{source}(m)=j\}\\
      &\quad +\frac{a_i\gamma_{m_0}}{2|\calS_{ig}[t]|}\sum_{m^{\prime}\in \calS_{ig}[t]}\mathbbm{1}\{\mathsf{source}(m^{\prime})=j\}+\frac{a_i(1-\gamma_{m_0})}{2|\calL_{ig}[t]|}\sum_{m^{\prime}\in \calL_{ig}[t]}\mathbbm{1}\{\mathsf{source}(m^{\prime})=j\}.
\end{align*}
Follow the same line as in the proof of case I, it can be shown that the above ${\bf M}_{ij}$ satisfies conditions 1), 2) and 3).

In case III, case IV, case V and case VI, at least one of $\calS_{ig}[t]$ and $\calL_{ig}[t]$ is empty, without loss of generality, assume that $\calS_{ig}[t]$ is empty. By the definition of $\calS_{ig}[t]$, we know that the set $\calM_{is}[t]$ is covered by $\calF$. On the other hand, by the definition of $\calM_{is}[t]$, a minimum cover of $\calM_{is}[t]$ is of size $f$. Since $|\calF|\le f$, then we know $\calF$ is a minimum cover of $\calM_{is}[t]$ and $|\calF|=f$.
 From the definition of $\calM_{is}[t]$, we know there exists a message with the smallest value in $\calM^*_i[t]$, denoted by $m_s$ is not covered by $\calF$. So, we can
  use singleton $\{m_s\}$ to mimic the role of $\calS_{ig}[t]$ in cases I and II. Similarly, we can use the same trick when $\calL_{ig}[t]$ is empty. The \emph{untampered message representation of $v_i[t]$} and \emph{message weight} are defined similarly as that in case I and case II.

To show the above constructions satisfy the last condition in Theorem \ref{claim_1}, we need the following claim.

\begin{claim}
\label{claimw}
For node $i\in \calV-\calF$, in the untampered message representation of $v_i[t]$, at most one of the sets $\calS_{ig}[t]$ and $\calL_{ig}[t]$ contains messages with assigned weights less than $\beta$, where $\beta=\frac{1}{16n^{2l}}$.
\end{claim}

\begin{proof}
An untampered message is either in $\calM^*_i[t]$ or in $\calS_{ig}[t]\cup\calL_{ig}[t]$.

For case V and case VI, both $\calS_{ig}[t]$ and $\calL_{ig}[t]$ are empty, all untampered messages are contained in $\calM^*_i[t]$.
For each untampered message in $\calM_i^*[t]$, its weight in the untampered message representation is $a_i=\frac{1}{|\calM^*_i[t]|+1}$. In $\calM_i[t]$, there are at most $n$ messages were transmitted via one hop, at most $n^2$ messages were transmitted via two hops. In general, $\calM_i[t]$ contains at most $n^d$ messages that were transmitted via $d$ hops, where $d$ is an integer in $\{1,\ldots, l\}$. Thus,
\begin{align*}
|\calM^*_i[t]|+1&\le |\calM_i[t]|\\
&\le n+n^2+\ldots+n^l\\
&=\frac{n(n^l-1)}{l^*}\\
&\overset{(a)}{\le} \frac{n(n^l-1)}{\frac{n}{2}}\\
&\le 2n^{l}.
\end{align*}
Inequality $(a)$ is true because $n\ge 2$. Thus, $a_i\ge \frac{1}{2n^l}$. In cases V and VI, as both $\calS_{ig}[t]$ and $\calL_{ig}[t]$ are empty, all untampered messages are with weight no less than $\frac{1}{2n^l}$.

For case III and case IV, WLOG, assume $\calS_{ig}[t]$ is empty.
An untampered message is either in $\calM_i^*[t]$ or in $\calL_{ig}[t]$. Since for each untampered message in $\calM_i^*[t]$, the weight assigned to it in the untampered message representation of $v_i[t]$ is at least $\frac{1}{2n^l}$. Thus, only $\calL_{ig}[t]$ may contain untampered messages with assigned weights less than $\frac{1}{2n^l}$.

For case II, both $\calS_{ig}[t]$ and $\calL_{ig}[t]$ are nonempty, an untampered message is in one of $\calM_i^*[t]$, $\calS_{ig}[t]$ and $\calL_{ig}[t]$. In the untampered message representation of $v_i[t]$,  either $\gamma_{m_0}\ge \frac{1}{2}$ or $1-\gamma_{m_0}\ge \frac{1}{2}$.  WLOG, assume that $\gamma_{m_0}\ge \frac{1}{2}$, which  implies that for each message in $\calS_{ig}[t]$, the assigned weight is at least $\frac{a_i}{4|\calS_{ig}[t]|}\ge \frac{1}{16n^{2l}}$, since $|\calS_{ig}[t]|\le |\calM_i[t]|\le 2n^l$. Let $\beta=\frac{1}{16n^{2l}}$, then we can conclude that only $\calL_{ig}[t]$ may contain untampered messages with assigned weights less than $\beta$.

It can be shown similarly that the above claim also holds for case I.

\end{proof}
Now we are ready to show the following property is also true.

\begin{claim}
\label{claimrg}
For any $t\geq 1$, there exists a reduced graph $\widetilde{G^l}_{\calF}\in R_\calF$ such that
$\beta \, {\bf{H}}[t] ~ \leq ~  {{\bf{M}}[t]}$.
\end{claim}
\begin{proof}
 We construct the desired reduced graph $\widetilde{G^l}_{\calF}$ as follows.
 Let
 \[
 E=\{e\in \calE(G^{l}):~\calV(P(e))\cap \calF\not=\O\}
 \]
  be the set of edges in $G^{l}$ that are covered by node set $\calF$.

For a fault-free node $i$: (i) if both $\calS_{ig}[t]$ and $\calL_{ig}[t]$ are empty, then choose $C_i=\O$; (ii) if one of $\calS_{ig}[t]$ and $\calL_{ig}[t]$ is empty, WLOG, assume that $\calS_{ig}[t]$ is empty, then choose $C_i=\calT^*(\calM_{il}[t])$; (iii) if both $\calS_{ig}[t]$ and $\calL_{ig}[t]$ are nonempty, WLOG, assume that the weight assigned to every message in $\calS_{ig}[t]$ is lower bounded by $\beta$, then choose $C_i=\calT^*(\calM_{il}[t])$. Let
\[
E_i=\{e\in \calE(G^{l}):~ e~\text{is an incoming edge of node $i$ in}~G^{l}~~ \text{and}~ \calV(P(e))\cap C_i\not=\O\}
\]
 be the set of incoming edges of node $i$ in $G^{l}$ that are covered by node set $C_i$.

Set $\calV(\widetilde{G^l}_{\calF})=\calV(G)-\calF$. And let $\calE(\widetilde{G^l}_{\calF})=\calE(\widetilde{G^l})-E-\cup_{i\in\calV-\calF}E_i$.

From claim \ref{claimw}, for node $i$, at most one of the sets $\calS_{ig}[t]$ and $\calL_{ig}[t]$ contains messages with assigned weights less than $\beta$. Then it is easy to see that the adjacency matrix of the obtained reduced graph, ${\bf{H}}[t]$, has the property that $ {\beta \,\bf{H}}[t] ~ \leq ~  {{\bf{M}}[t]}$.

\end{proof}

\end{proof}

\subsection{Correctness of Algorithm 1}\label{app:correctness}

\begin{lemma}
\label{l_product_H}
In the product below of ${\bf H}[t]$ matrices for consecutive
$\tau(n-\phi)$ iterations, at least one column is non-zero.
\[
\Pi_{t=z}^{z+\tau(n-\phi)-1} \, {\bf H}[t]
\]
\end{lemma}
\begin{proof}
Since the above product consists of $\tau(n-\phi)$ matrices
in $R_F$,
at least one of the $\tau$ distinct connectivity matrices
in $R_\calF$, say matrix ${\bf H}_*$, will appear in the above
product at least $n-\phi$ times.

Now observe that: (i)
By Lemma \ref{l_one_column}, ${\bf H}_*^{n-\phi}$ contains a non-zero
column, say the $k$-th column is non-zero,
and (ii) all the ${\bf H}[t]$ matrices in the product contain a non-zero diagonal.
These two observations together imply that the $k$-th column in the above product
is non-zero.
\end{proof}

Let us now define a sequence of matrices ${\bf Q}(i)$ such that
each of these matrices is a product of $\tau(n-\phi)$ of the
${\bf M}[t]$ matrices. Specifically,
\[
{\bf Q}(i) ~=~ \Pi_{t=(i-1)\tau(n-\phi)+1}^{i\tau(n-\phi)} ~ {\bf M}[t]
\]
Observe that
\begin{eqnarray}
{\bf v}[k\tau(n-\phi)] & = & \left(\, \Pi_{i=1}^k ~ {\bf Q}(i) \,\right)~{\bf v}[0]
\end{eqnarray}

\begin{lemma}
\label{l_Q}
For $i\geq 1$, ${\bf Q}(i)$ is a scrambling row stochastic matrix,
and $\lambda({\bf Q}(i))$ is bounded from above by a constant
smaller than 1.
\end{lemma}
\begin{proof}

${\bf Q}(i)$ is a product of row stochastic matrices (${\bf M}[t]$), therefore,
${\bf Q}(i)$ is row stochastic.

From Lemma \ref{claimrg}, for each $t$,
\[
\beta \, {\bf H}[t] ~ \leq ~ {\bf M}[t]
\]
Therefore,
\[
\beta^{\tau(n-\phi)} ~ \Pi_{t=(i-1)\tau(n-\phi)+1}^{i\tau(n-\phi)} ~ {\bf H}[t] ~ \leq
~ {\bf Q}(i)
\]
By using $z=(i-1)(n-\phi)+1$ in Lemma \ref{l_product_H},
we conclude that the matrix product on the left side
of the above inequality contains a non-zero column. Therefore, ${\bf Q}(i)$ contains
a non-zero column as well. Therefore, ${\bf Q}(i)$ is a scrambling matrix.

Observe that $\tau(n-\phi)$ is finite, therefore, $\beta^{\tau(n-\phi)}$
is non-zero. Since the non-zero terms in ${\bf H}[t]$ matrices are all 1,
the non-zero entries in $\Pi_{t=(i-1)\tau(n-\phi)+1}^{i\tau(n-\phi)} {\bf H}[t]$
must each be $\geq$ 1. Therefore, there exists a non-zero column in ${\bf Q}(i)$
with all the entries in the column being $\geq \beta^{\tau(n-\phi)}$.
Therefore $\lambda({\bf Q}(i))\leq 1-\beta^{\tau(n-\phi)}$.
\end{proof}

\begin{proof}[Proof of Theorem \ref{t}]

Since ${\bf v}[t]={\bf M}[t]\,v[t-1]$, and ${\bf M}[t]$ is a row stochastic matrix, it
follows that
\emph {Algorithm 1} satisfies the validity condition.

By Claim \ref{claim_delta},
\begin{eqnarray}
\lim_{t\rightarrow \infty} \delta(\Pi_{i=1}^t {\bf M}[t])
& \leq & \lim_{t\rightarrow\infty} \Pi_{i=1}^t \lambda({\bf M}[t]) \\
& \leq & \lim_{i\rightarrow\infty} \Pi_{i=1}^{\lfloor\frac{t}{\tau(n-\phi)}\rfloor} \lambda({\bf Q}(i)) \\
& = & 0
\end{eqnarray}
The above argument makes use of the facts that
$\lambda({\bf M}[t])\leq 1$ and $\lambda({\bf Q}(i))\leq (1-\beta^{\tau(n-\phi)})<1$.
Thus, the rows of $\Pi_{i=1}^t {\bf M}[t]$ become identical in the limit.
This observation, and the fact that ${\bf v}[t]=(\Pi_{i=1}^t {\bf M}[i]){\bf v}[t-1]$ together imply that
the state of the fault-free nodes satisfies the
convergence condition.

Now, the validity and convergence conditions
together imply that
there exists a positive scalar $c$ such that
\[
\lim_{t\rightarrow\infty}
{\bf v}[t] ~ = ~ \lim_{t\rightarrow\infty} \left( \Pi_{i=1}^t {\bf M}[i]) \right)\,
{\bf v}[0] ~ = ~ c\,{\bf 1}
\]
where {\bf 1} denotes a column with all its entries being 1.
\end{proof}

\section{Connection to existing work}\label{app:connection}
\subsection{Undirected graph when $l\ge l^*$}
\begin{proof}[Proof of Theorem \ref{equiUndirected}]
First we show ``Condition NC implies $n\ge 3f+1$ and node connectivity at least $2f+1$". It has already been shown in Corollary \ref{cdegree} that $n\ge 3f+1$. It remains to show the node connectivity of $G$ is at least $2f+1$.
We prove this by contradiction. Suppose the node-connectivity is no more than $2f$. Let $S$ be a min cut of $G$, then $|S|\le 2f$.
Let $K_1$ and $K_2$ be two connected components in $G_S$, the subgraph of $G$ induced by node set $\calV(G)-S$.

Construct a node partition of $G$ as follows:
Let $L=K_1, R=K_2$ and $C=\calV-F-L-R$, where (1) if $|S|\ge f+1$, let $F\subseteq S$ such that $|F|=f$; (2) otherwise, let $F=S$.  For the later case, there is no path between $L\cup C$ and $R$ in $G_F$, then $\kappa(L\cup C, i)=0\le f$ for any $i\in R$ in $G_F$. Similarly, $\kappa(R\cup C, j)=0\le f$ for any $j\in L$. On the other hand, we know that $G$ satisfies Condition NC. Thus, we arrive at a contradiction.

For the former case, i.e., $F\subset S$, since $G$ satisfies Condition NC, WLOG, assume $R\cup C\Rightarrow_{l^*} L$ in $G_F$, i.e., there exists a node $i\in L$ such that there are at least $f+1$ disjoint paths from set $R\cup C$ to node $i$ in $G_F$. Add an additional node $y$ and connect node $y$ to all nodes in $R\cup C$.
Denote the resulting graph by $G_F^{\prime}$.  From Menger's Theorem we know that a min $y, i$-cut in graph $G_F^{\prime}$ has size at least $f+1$. On the other hand, since $S$ is a cut of $G$, then we know $S-F$ is a $y, i$--cut in $G_F^{\prime}$. In addition, we know $|S-F|=|S|-|F|\le 2f-f\le f$. Thus we arrive at a contradiction.

Next we show that ``$n\ge 3f+1$ and $2f+1$ node-connectivity also imply Condition NC". Consider an arbitrary node partition $L, R, C, F$ such that $L\not=\O, R\not=\O$ and $|F|\le f$. Since $n\ge 3f+1$ and $|F|\le f$, either $|L\cup C|\ge f+1$ or $|R\cup C|\ge f+1$. WLOG, assume that $|R\cup C|\ge f+1$. Add a node $y$ connecting to all nodes in $R\cup C \cup F$ and denote the newly obtained graph by $G^{\prime\prime}$. By Expansion Lemma\footnote{\textbf{Expansion Lemma}: If $G$ is a $k$-connected graph, and $G^{\prime}$ is formed from $G$ by adding a vertex $y$ having at least $k$ neighbors in $G$, then $G^{\prime}$ is $k$-connected. }, $G^{\prime\prime}$ is $|F|+f+1$ connected. Thus, fix $i\in L$. There are at least $|F|+f+1$ internally disjoint $y, i$--paths. So there are at least $f+1$ internally disjoint $y, i$--paths in $G^{\prime\prime}_F$. Thus $R\cup C\Rightarrow_{l^*} L$ in $G_F$. Since this holds for all partitions of the form $L, R, C, F$ where $L\not=\O$, $R\not=\O$ and $|F|\le f$, then we conclude that Condition NC holds. This completes the proof.

\end{proof}

\subsection{Directed graph when $l\ge l^*$}

We first state the alternative condition of Condition 1.
\begin{definition}
Given disjoint subsets $A, B, F$ of $\calV(G)$ such that $|F|\le f$, set $A$ is said to propagate in $\calV-\calF$ to set $B$ if either (i) $B=\O$, or (ii) for each node $b\in B$, there exist at least $f+1$ disjoint $(A, b)$--paths excluding $F$.
\end{definition}

We will denote the fact that set $A$ propagates in $\calV-\calF$ to set $B$ by the notation
\begin{align*}
A\overset{\calV-\calF}{\rightsquigarrow}B.
\end{align*}

When it is not true that $A\overset{\calV-\calF}{\rightsquigarrow}B$, we will denote that fact by
\begin{align*}
A\overset{\calV-\calF}{\not\rightsquigarrow}B.
\end{align*}




\begin{theorem}
Given graph $G$, for any node partition $A, B, F$ of $\calV$, where $A$ and $B$ are both non-empty, and $|F|\le f$, then either $A\overset{\calV-F}{\rightsquigarrow}B$ or $B\overset{\calV-F}{\rightsquigarrow}A$ holds $\iff$ for any partition $L, C, R, F$ of $\calV$, such that both $L$ and $R$ are non-empty, and $|F|\le f$, either $L\cup C\to R$, or $R\cup C\to L$.
\end{theorem}

For ease of future reference, we term the first condition in the above theorem as Condition Propagate.

%
%
%


\begin{proof}[Proof of Theorem \ref{equivaDirected}]
We first show that Condition NC implies Condition 1.

For any node partition $L, C, R, F$ of $G$ such that $L\not=\O, R\not=\O$ and $|F|\le f$, in the induced subgraph $G_F$,
at least one of the two conditions below must be true: (i) $R\cup C\Rightarrow_l L$; (ii) $L\cup C\Rightarrow_l R$. Without loss of generality, assume that $R\cup C\Rightarrow_l L$ and node $i\in L$ has at least $f+1$ disjoint paths from $R\cup C$. For each such path, there exist at least an edge that goes from $R\cup C$ to a node in $L$. Since all the paths considered are disjoint, thus $R\cup C$ contains at least $f+1$ incoming neighbors of $L$.

We next show that Condition Propagate implies Condition NC. We prove this by contradiction. Suppose, on the contrary, that Condition NC does not hold. There exists a partition $L, C, R, F$ of $G$ such that $L\not=\O, R\not=\O$ and $|F|\le f$, in the induced subgraph $G_F$,
 (i) $R\cup C\not\Rightarrow_l L$; (ii) $L\cup C\not\Rightarrow_l R$. For each node $i$ in $L$, there are at most $f$ disjoint $(R\cup C, i)$ paths excluding $F$. Thus $R\cup C \overset{\calV-\calF}{\not\rightsquigarrow}L$.

On the other hand, as $L\cup C\not\Rightarrow_l R$, for each node $j\in R$, there are at most $f$ disjoint paths from $L\cup C$ to $j$ excluding $F$, which further implies that there are at most $f$ disjoint paths from $L$ to $j$ excluding $F$.
Thus, $L \overset{\calV-\calF}{\not\rightsquigarrow}R\cup C$. This contradicts the assumption that Condition Propagate holds. Thus we conclude that Condition Propagate implies Condition NC.

In addition, we know Condition Propagate $\iff$ Condition 1. Therefore, Condition NC $\iff$ Condition Propagate $\iff$ Condition 1.
\end{proof}


%
%
%
%
%
%
%
%
%
%
%
%

%
%


\end{document}